\documentclass[11pt]{article}

\usepackage[utf8]{inputenc}
\usepackage[T1]{fontenc}
\usepackage[margin=2.5cm]{geometry}
\usepackage{setspace}
\usepackage{amsmath}
\usepackage{amssymb}
\usepackage{amsthm}
\usepackage{booktabs}
\usepackage{graphicx}
\usepackage{makecell}
\usepackage{enumerate}
\usepackage{natbib}
\usepackage[linesnumbered,ruled,vlined]{algorithm2e}
\usepackage{comment}
\usepackage{placeins}
\usepackage{xurl}
\usepackage{hyperref}
\hypersetup{
  pdftitle={Polynomial and Pseudopolynomial Algorithms for Two Classes of Bin Packing Instances},
  pdfauthor={Renan Fernando Franco da Silva; Vinícius Loti de Lima; Rafael C. S. Schouery; Jean-François Côté; Manuel Iori}}
\SetKwInOut{KwIn}{Input}
\SetKwInOut{KwOut}{Output}

\SetKw{KwDownTo}{downto}
\SetAlFnt{\small}
\SetAlCapFnt{\small}
\SetAlgoNlRelativeSize{-1}
\newtheorem{lemma}{Lemma}
\newtheorem{theorem}{Theorem}
\newtheorem{remark}{Remark}
\newtheorem{definition}{Definition}
\newtheorem{observation}{Observation}
\newtheorem{property}{Property}

\title{Polynomial and Pseudopolynomial Algorithms for Two Classes of Bin Packing Instances}
\author{%
Renan Fernando Franco da Silva\textsuperscript{1},
Vin\'{i}cius Loti de Lima\textsuperscript{2},
Rafael C. S. Schouery\textsuperscript{1},\\
Jean-Fran\c{c}ois C\^{o}t\'{e}\textsuperscript{3},
Manuel Iori\textsuperscript{4}\\[0.5em]
\small \textsuperscript{1}Institute of Computing, University of Campinas, Campinas, SP, Brazil\\
\small \textsuperscript{2}MMPROS, Amazon, Bellevue, WA, USA\\
\small \textsuperscript{3}Facult\'{e} des Sciences de l'administration, Universit\'{e} Laval, Qu\'{e}bec, QC, Canada\\
\small \textsuperscript{4}DISMI, University of Modena and Reggio Emilia, Reggio Emilia, Italy
}
\date{}

\begin{document}
\onehalfspacing
\maketitle

\begin{abstract}
The Bin Packing Problem (BPP) and Cutting Stock Problem (CSP) are classical combinatorial optimization problems extensively studied since the 1960s. State-of-the-art exact algorithms are based on set-cover and arc-flow models whose linear relaxation, rounded up, matches the integer optimum for most benchmark instances, a condition known as the Integer Round-up Property (IRUP). In 2016, Delorme et al. showed that most existing instances could be solved within ten minutes by approaches exploiting this property. This motivated them to introduce two new classes of instances, Augmented IRUP (AI) and Augmented Non-IRUP (ANI), designed to challenge the exact algorithms available at the time. Although these classes have motivated significant advances over the past decade, 13 out of 500 AI and ANI instances remain unsolved within standard time limits in the literature. In this paper, we show that while AI and ANI are particularly hard for mixed-integer programming methods, the BPP restricted to these classes is not strongly NP-hard. We present polynomial-time algorithms for the AI class and pseudopolynomial-time algorithms for the ANI class, which solve all such instances orders of magnitude faster than previous approaches. They are also straightforward to adapt to the Skiving Stock Problem, the dual counterpart of the CSP. In addition, they can be used as preprocessing routines in exact methods, as their runtime is independent of the instance class, although they are guaranteed to return an optimality status only for instances belonging to the class for which they were designed.
\end{abstract}


\section{Introduction}

The \emph{Bin Packing Problem}~(BPP) and the \emph{Cutting Stock Problem}~(CSP) are classical combinatorial optimization problems with numerous applications in manufacturing, logistics, and resource management (see, e.g.,~\citealp{Hessler2022,Martin_2023}). In the BPP, a set $I = \{1, \ldots, n\}$ of items with sizes $w_i \in \mathbb{Z}_+$ must be packed into bins of capacity $W \in \mathbb{Z}_+$ without exceeding capacity, while minimizing the number of bins used. In the CSP, a set $I = \{1, \ldots, n\}$ of item types is given, where each $i \in I$ has size $w_i \in \mathbb{Z}_+$ and demand $d_i \in \mathbb{Z}_+$. The goal is to cut $d_i$ copies of each type from rolls of length $W \in \mathbb{Z}_+$ so that the total size assigned to any roll does not exceed $W$, while minimizing the number of rolls used. The BPP is a special case of the CSP where $d_i = 1$ for all $i \in I$.

Although the BPP is typically associated with packing and the CSP with cutting, the operations of packing and cutting are mathematically equivalent. In particular, the BPP admits a polynomial-time reduction to the CSP, whereas the converse reduction is only pseudopolynomial in the item demands. In practice, CSP instances often involve items with large demands, whereas BPP instances typically contain fewer items with identical weights (see, e.g., \citealp{Wascher_2007}). Due to their computational complexity and practical relevance, both problems have been extensively studied, resulting in a wide range of exact and heuristic solution methods (see, e.g., \citealp{Delorme_2016}).

In the Skiving Stock Problem~(SSP), one is given a set $I = \{1, \ldots, n\}$ of item types, where each item type $i \in I$ has a size $w_i\in \mathbb{Z}_+$ and multiplicity $m_i \in \mathbb{Z}_+$\@. The goal is to merge these items to form rolls of length at least $W \in \mathbb{Z}_+$, such that each item type $i$ is used at most $m_i$ times, while maximizing the number of rolls produced. The SSP is commonly regarded as the dual counterpart of the CSP (see, e.g., \citealp{Martinovic_2016}).

Proposing new benchmark instances for the BPP that are both reasonably sized and computationally challenging has become increasingly difficult due to the availability of strong formulations, such as set covering and arc-flow models~(see, e.g., \citealp{deLima2022}), as well as the continuous advancement of solution methods. When proposing new instances, authors have usually taken into consideration the following property.

\begin{property}
	An instance is said to possess the Integer Rounding Up Property (IRUP) for a given formulation if the ceiling of the optimal value of its linear relaxation, $\lceil z_\text{LP} \rceil$, is equal to the optimal value $z_\text{ILP}$ of the corresponding integer problem. Instances are classified as IRUP ($\lceil z_\text{LP} \rceil = z_\text{ILP}$) or Non-IRUP ($\lceil z_\text{LP} \rceil < z_\text{ILP}$) depending on whether they satisfy this property or not.
\end{property}

Within this field, \citet{Delorme_2016} proposed two benchmark classes, known as the Augmented IRUP (AI) and Augmented Non-IRUP (\textsc{ANI}) classes, each containing 250 instances. These instances have guided the development of new algorithmic techniques, and significant progress has indeed been achieved since their proposal. Notably, state-of-the-art algorithms~(culminating in~\citealp{Silva_2025}) are now able to solve almost all such instances within a one-hour time limit, with only 13 instances remaining unsolved. Both classes have also been used as challenging SSP instances in the literature.

In this paper, we investigate the BPP, CSP, and SSP, with a particular focus on the \textsc{AI} and \textsc{ANI} benchmark classes. Our main contributions are as follows:
\begin{itemize}

\item We study the structure of the \textsc{AI} and \textsc{ANI} classes and introduce generalized versions of their construction. We establish structural properties, including additional properties satisfied by all benchmark instances reported in the literature, which we exploit to develop efficient algorithms;

\item We show that the three studied problems restricted to the \textsc{AI} class can be solved in polynomial time, whereas those restricted to the \textsc{ANI} class can be solved in pseudopolynomial time. For both classes, we first derive theoretically valid but impractical algorithms and then develop substantially more efficient ones, with empirical running times close to $O(n^3)$ for \textsc{AI} and $O(n^2W)$ for \textsc{ANI};

\item Computational experiments show that our algorithms solve all 500 \textsc{AI} and \textsc{ANI} benchmark instances to proven optimality, while requiring only a few seconds per instance.
\end{itemize}

The remainder of this article is organized as follows. Section~\ref{sec::review} reviews the related literature on benchmarks and exact algorithms. Section~\ref{sec::instances} describes some interesting properties about AI and ANI benchmarks. Sections~\ref{sec::AI} and~\ref{sec::ANI} introduce the polynomial-time and pseudopolynomial-time algorithms for the \textsc{AI} and \textsc{ANI} classes, respectively. Section~\ref{sec::experiments} reports computational experiments, and Section~\ref{sec::conclusion} presents concluding remarks and directions for future research.

\section{Literature Review}~\label{sec::review}
Results reported in the literature for the CSP and the BPP are largely interchangeable. The main difference is that approaches explicitly tailored to the CSP are, in general, better suited for instances with large item demands. Moreover, since the literature on the BPP and the CSP is much richer than the one on the SSP, and algorithms for the SSP are often adaptations of methods originally developed for the BPP and the CSP, in this section, we focus primarily on the BPP and CSP, commenting on the SSP only when necessary.

Surveys dedicated to approximation algorithms for bin packing problems have been published since the 1980s~\citep{Garey_1981,Coffman_1984}, with the most recent one being proposed by \citet{Coffman_2013}. A problem typology was originally presented by~\citet{Dyckhoff_1990}, and later extended by~\citet{Wascher_2007}. The most recent survey addressing both the BPP and the CSP is due to~\citet{Delorme_2016}, who focused on mathematical formulations and presented computational experiments with the most relevant solvers available at the time. Since its publication, however, substantial progress has been achieved in the development of exact algorithms. In what follows, we therefore review the most relevant results on exact methods, referring to them for simplicity as results for the BPP.

For several decades, the most effective exact algorithms for the BPP were based on Combinatorial Branch-and-Bound (CB\&B) methods, that is, branch-and-bound approaches not relying on Linear Programming (LP). The most influential algorithm in this category is MTP, proposed by~\citet{Martello_1990}. MTP combines multiple lower bounds, reduction procedures, and effective heuristics, and it became a standard reference for the exact solution of the BPP, in part due to the public availability of its implementation. Later,~\citet{Scholl_1997} proposed BISON, which represented the state-of-the-art among CB\&B algorithms for the BPP at that time. BISON incorporates several powerful components of MTP, strengthens them with new lower bounds, and integrates techniques such as Tabu Search.

With the availability of increasingly efficient LP solvers, branch-and-price (B\&P) approaches gradually replaced CB\&B methods. \citet{Vance_1994} and \citet{Vance_1998} are two notable works from the beginning of this transition period, proposing B\&P algorithms based on the classical Set Covering Formulation (SCF) for the BPP and the CSP, respectively. The classical SCF for the CSP (see \citealp{Uchoa_2026} for a detailed discussion on its origin) is given by:
\begin{align}
\min \quad & \sum_{P \in \mathcal{P}} x_P \label{scf:obj} \\
\text{s.t.} \quad & \sum_{P \in \mathcal{P}} a_i^P x_P \;\ge\; d_i, \quad & & \forall i \in I, \label{scf:dem} \\
& x_P \in \mathbb{Z}_+, \quad & & \forall P \in \mathcal{P}. \label{scf:int}
\end{align}
Here, \( \mathcal{P} \) denotes the set of all feasible cutting patterns, that is, all multisets of items whose total weight does not exceed \( W \), and \( a_i^P \) denotes the number of items of type \( i \) in pattern \( P \). Each variable \( x_P \) represents the number of times pattern \( P \) is used, and the objective function minimizes the total number of rolls. The set of patterns \( \mathcal{P} \) may be restricted to include only patterns in which each item \( i \) appears at most \( d_i \) times, and such patterns are called \emph{proper} patterns. Patterns that do not satisfy this condition are referred to as \emph{non-proper} patterns.

By replacing the inequalities \eqref{scf:dem} with equalities, we obtain the well-known Set Partitioning Formulation (SPF) for the CSP. Moreover, the SPF with a maximization objective and a set of patterns consisting of multisets of items whose weight sum is at least \( W \) models the SSP.

Subsequently,~\citet{Carvalho_1999} proposed a B\&P for the BPP relying on the arc-flow formulation (AFF), which can be obtained from the SCF by interpreting patterns as paths in a specific network and then decomposing these paths into arcs by introducing a variable for each arc.

Both SCF and AFF are known to yield strong linear relaxations. Indeed, a central conjecture in the context of the CSP is the following.

\begin{property}
The Modified Integer Rounding Up Property (MIRUP) conjecture (see, e.g., \citealp{Scheithauer_1995}) states that every CSP instance satisfies
\[
	z_\mathrm{ILP} \leq \lceil z_\mathrm{LP} \rceil + 1,
\]
where \( z_\mathrm{ILP} \) denotes the optimal value of the integer problem and \( z_\mathrm{LP} \) is the optimal value of the linear relaxation of the SCF allowing non-proper patterns.
\end{property}

While the AI and ANI benchmark instances satisfy the MIRUP inequality, Non-IRUP instances exist even for the SCF restricted to proper patterns. Furthermore, since the SCF with non-proper patterns and the AFF are equivalent~\citep{Delorme_2020}, the conjecture applies to the AFF as well. Hereafter, instances are classified as IRUP or Non-IRUP with respect to the SCF with non-proper patterns.

Several additional LP-based exact solvers were proposed in subsequent years. Notably,~\citet{Belov_2006} introduced a branch-cut-and-price algorithm, named \texttt{BELOV}, based on the SCF, which proved to be extremely powerful. In their survey,~\citet{Delorme_2016} reviewed exact methods available at that time and observed that the benchmark instances commonly used in previous studies were no longer sufficient to distinguish among the leading solvers. Indeed, state-of-the-art solvers could handle these instances very efficiently, with \texttt{BELOV} being able to solve all of them within a few minutes when executed on modern hardware with a recent version of the CPLEX solver.

In response to this situation,~\citet{Delorme_2016} proposed two new benchmark classes, Augmented Non-IRUP (\textsc{ANI}) and Augmented IRUP (\textsc{AI}), which proved to be particularly challenging for all exact solvers available at that time, as well as for those developed in the following decade. Consequently, subsequent research has largely focused on improving performance on these classes.

These instances are generated using a specific construction technique and exhibit distinctive properties. For \textsc{ANI} instances, it holds that \( z_\mathrm{LP} \in \mathbb{Z}_+ \) and \( z_\mathrm{ILP} = z_\mathrm{LP} + 1 \), which implies that they are Non-IRUP, and this property holds even when only proper patterns are used. On the other hand, \textsc{AI} instances satisfy \( z_\mathrm{ILP} = z_\mathrm{LP} \in \mathbb{Z}_+ \), meaning that they are IRUP. Recent works by \citet{Wei_2020}, \citet{Pessoa_2020}, \citet{Loti_2023}, \citet{Baldacci_2024}, and \citet{Silva_2025} have reported significant progress on these two benchmarks, with \citet{Silva_2025} successfully solving 487 out of 500 instances.

As discussed by~\citet{Silva_2025}, the main difficulty of these benchmarks lies in solving their LP relaxations (from SCF or AFF), primarily due to numerical and stability issues. These issues are largely caused by the high degree of symmetry introduced by the particular structure of the instances (see Section~\ref{sec::instances} for more details), which makes the LP relaxations significantly harder to solve than those of other instance classes. The symmetries were so drastic that they completely degraded the convergence of the LP relaxation solvers, particularly in column generation approaches. Although solving such relaxations is a weakly NP-hard problem, it took nearly a decade for the literature to develop sophisticated techniques capable of producing effective algorithms. Once an effective algorithm for the LP relaxation is available, most of these classes are no longer particularly difficult. In fact, as noted by~\citet{Silva_2025}, around 70\% of the AI and ANI instances can be solved easily with only a few branching decisions or subset-row cutting planes.

In this work, we propose a polynomial-time exact method for \textsc{AI} instances and a pseudopoly\-no\-mial-time exact method for the \textsc{ANI} instances. These two methods can solve all instances from both classes in a few seconds per instance, which is orders of magnitude faster than previous state-of-the-art solvers. In light of these findings, we recognize that these classes played an important role in the development of better solvers for the BPP and CSP, but we believe that future research should focus on devising other more challenging instances.

\section{Some Relevant Properties of the AI and ANI Classes}~\label{sec::instances}
To build an \textsc{ANI} instance for the BPP, \citet{Delorme_2016} propose an algorithm designed to increase the likelihood of generating, in practice, an instance with a specific number of items. For a given $h \in \mathbb{Z}_+$ and $n = 15 + 3h$, an \textsc{ANI} instance $\mathcal{I} = \{1, \dots, n\}$ with bin capacity $W$ is constructed by starting from a Non-IRUP instance $O$ containing 15 items such that
\[
	\sum_{i \in O} w_i = 3W, \qquad z_{\mathrm{LP}}^O = 3, \qquad z_{\mathrm{ILP}}^O = 4.
\]
We refer to $O$ as the \emph{original instance}. The original instances
used by \citet{Delorme_2016} are the ones proposed by~\citet{Caprara_2015}. The construction then proceeds for $h$ iterations, each of which expands the instance by adding a triplet of items, referred to as an \emph{original triplet}. Specifically, at iteration $k \in \{1,\dots,h\}$, an original triplet $(\ell_k,p_k,q_k)$ with total weight $W$ is added, with its first item $\ell_k$ chosen such that
\begin{equation}
    \label{cond:no-completion}
    \nexists S \subseteq
    \{\ell_g,p_g,q_g : 1 \leq g \leq k-1\} \cup O
    \text{ such that }
    \sum_{i \in S} w_i + w_{\ell_k} = W.
\end{equation}
Hereafter, whenever the indexed symbols $\ell_k$, $p_k$, and $q_k$ are used, they refer to the items of the $k$-th original triplet, with $\ell_k$ denoting the item that satisfies~\eqref{cond:no-completion}. Equivalently, at iteration~$k$, every pattern with no waste that contains $\ell_k$ also contains either $p_k$ or $q_k$ (or both). Hence, this construction yields an \textsc{ANI} instance \mbox{$\mathcal{I} = O \cup \{\ell_k, p_k, q_k\colon k \in \{1, \dots, h\}\}$} that remains Non-IRUP, with \(z_{\mathrm{LP}}^{\mathcal{I}} = h + 3\) and \(z_{\mathrm{ILP}}^{\mathcal{I}} = h + 4\). Each \textsc{AI} instance \(\mathcal{I}'\) is derived from an \textsc{ANI} instance \(\mathcal{I}\) by selecting an arbitrary item \(o \in O\) and splitting it into two items \(o_1\) and \(o_2\), in such a way that the resulting original instance \(O'\) has 16 items and satisfies \(z_{\mathrm{LP}}^{O'} = z_{\mathrm{ILP}}^{O'} = 3\) (i.e., the Non-IRUP property of $O$ is lost), while the associated \textsc{AI} instance $\mathcal{I}'$ fulfills \(z_{\mathrm{LP}}^{\mathcal{I}'} = z_{\mathrm{ILP}}^{\mathcal{I}'} = h + 3\). Observe that all optimal solutions for the \textsc{AI} class are \emph{perfect packing solutions}, that is, their patterns have no waste.

Next, we introduce a generalized construction process that does not restrict the original instance \(O\) to have 15 items, nor require extending \(O\) with tuples of size 3. Given constants $\alpha, \beta, \gamma \in \mathbb{Z}_+$ with $\gamma \geq 3$, a generalized $(\alpha, \beta, \gamma)$-\textsc{ANI} instance $\widetilde{\mathcal{I}}$ can be constructed by starting from an original instance \(\widetilde O\) with \(|\widetilde O| = \alpha\), satisfying \(z_{\mathrm{LP}}^{\widetilde O} = \sum_{i \in \widetilde O} w_i / W = \beta\) and \(z_{\mathrm{ILP}}^{\widetilde O} = \beta + 1\), and extending it by adding patterns composed of $\gamma$ items, i.e., $\gamma$-tuples, each of weight sum \(W\), where the first item $\ell_k$ in each tuple cannot be completed to a pattern with no waste using any subset of previously introduced items. A generalized $(\alpha, \beta, \gamma)$-\textsc{AI} instance $\widetilde{\mathcal{I}}'$ is then obtained from $\widetilde{\mathcal{I}}$ by selecting one item of \(\widetilde O\) and splitting it into two items, so that the resulting original instance \(\widetilde O'\) satisfies the IRUP. To the best of our knowledge, there is no guarantee that this transformation is possible for every Non-IRUP instance with an integrality gap equal to \(1\). However, \(\widetilde O\) can be restricted to instances that admit such a transformation. In this article, we also analyze the computational complexity of our algorithms for these generalized classes under the assumption stated in the following remark.

\begin{remark}\label{remark:residual-gap}
For every optimal solution $x$ of the SCF relaxation of an original instance $\widetilde O$, let $\widehat O$ denote the residual instance induced by $x-\lfloor x\rfloor$, i.e., after excluding all items covered by patterns with value~$1$. Throughout this article, we assume that every such residual instance $\widehat O$ has integrality gap~$1$, i.e., that it satisfies the MIRUP conjecture. This assumption follows immediately if the MIRUP conjecture holds in general and is also guaranteed when $\beta\leq 19/4$~\citep{Rietz_2002}. Since the AI and ANI instances proposed by~\citet{Delorme_2016} satisfy the latter condition, this assumption imposes no additional restriction on these classes.
\end{remark}

Each AI and ANI instance, including the generalized ones, can also be viewed as an SSP instance. Analogous statements hold for the SPF for both the CSP and the SSP. In particular, the sets of optimal linear relaxation solutions for SPF are identical, in both problems, for the generalized \textsc{AI} and \textsc{ANI} classes. For the generalized \textsc{AI} class, the optimal integer solutions coincide for the CSP and the SSP, whereas for the generalized \textsc{ANI} class, the optimal integer value of the SSP equals \( h + 2 \) (resp. $h + \beta - 1$), which corresponds to an absolute integrality gap of one unit, as in the CSP case.

Perfect packing solutions play a central role in generalized \textsc{AI} instances, where such solutions must be found, and in generalized \textsc{ANI} instances, where it must be shown that no perfect packing exists. Any perfect packing solution uses only patterns with no waste, which we refer to as \emph{full patterns}. We also define a \emph{full triplet} as a set of three items $(\ell,p,q)$ such that $w_\ell + w_p + w_q = W$, that is, a full pattern composed of exactly three items.

The following remark establishes a necessary property induced by condition~\eqref{cond:no-completion}.

\begin{remark}~\label{remark:distinct_weights}
	The set of items $\{\ell_1, \ldots, \ell_h\}$ has pairwise distinct weights, that is, \(w_{\ell_k} \neq w_{\ell_g}\) for all distinct \(k, g \in \{1, \ldots, h\}\).
\end{remark}

Moreover, all AI and ANI instances reported in the literature satisfy the following condition:
\begin{equation}\label{cond:a_is_large}
w_{\ell_k} \geq \frac{W}{2} \quad \text{for all } k \in \{1, \ldots, h\},
\end{equation}
which we later exploit to derive a more efficient algorithm. Note that condition~\eqref{cond:a_is_large} is not stated explicitly in~\citet{Delorme_2016}. Indeed, it is possible to generate instances in which some items $\ell_k$ have weight $w_{\ell_k}<W/2$, while preserving the same base instances, number of items, and bin capacity as those reported in the literature. However, in our experiments, even when explicitly enforcing the generation of items $\ell_k$ with weight $w_{\ell_k} < W/2$, we were typically able to obtain only two or three such items per instance. This suggests that instances satisfying \eqref{cond:a_is_large} are significantly more likely to occur during the generation process, which explains why all instances reported in the literature satisfy it.

We next introduce a few structural notions that will be useful for characterizing the instances considered in our algorithms. First, we define perfect-packing candidate instances, a property that every generalized \textsc{AI} and generalized \textsc{ANI} instance possesses. Then, we introduce notions of irreducibility for general perfect-packing candidate instances and original-irreducibility for generalized $(\alpha, \beta, \gamma)$-\textsc{ANI} instances, respectively.

\begin{definition}
An instance is called a \emph{perfect-packing candidate} if each item belongs to at least one full pattern.
\end{definition}

\begin{definition}\label{definition::irreducible_general}
A perfect-packing candidate instance $I$ is said to be \emph{reducible} if there exist distinct items $i,j \in I$ such that $\mathcal{P}_i^{\mathrm{full}} \subseteq \mathcal{P}_j^{\mathrm{full}}$, where $\mathcal{P}_i^{\mathrm{full}}$ and $\mathcal{P}_j^{\mathrm{full}}$ denote the sets of all full patterns containing $i$ and $j$, respectively. In this case, $(i,j)$ is called a \emph{reducible pair}. Conversely, $I$ is said to be \emph{irreducible} if there is no such pair.
\end{definition}

\begin{definition}\label{definition::irreducible}
An $(\alpha, \beta, \gamma)$-ANI instance $\widetilde{\mathcal{I}}$ is said to be
\emph{original-irreducible} (respectively, \emph{original-reducible})
if its corresponding original instance $\widetilde O$ is irreducible
(respectively, reducible).
\end{definition}

Remark~\ref{rmk:reducible_exist} shows that reducible instances exist. Another interesting structural property of \textsc{ANI} instances reported in the literature is stated in Remark~\ref{remark::irreducible_claim}, which is exploited by our algorithms in Section~\ref{sec::ANI}. Analogous original-irreducibility properties could also be stated for \textsc{AI} instances, although they are not required by our algorithms.

\begin{remark}\label{rmk:reducible_exist}
There exist perfect-packing candidate instances that are reducible. For example, consider an irreducible perfect-packing candidate instance $I$ with bin capacity $W$, such as an ANI instance. We can construct a reducible instance from $I$ by multiplying every item weight and the bin capacity by $2$. Then, select an item $i$, which now has an even weight, and split it into two items $b$ and $c$, each with an odd weight and satisfying $w_b + w_c = w_i$. The resulting instance is reducible, and $(b, c)$ forms a reducible pair.
\end{remark}

\begin{remark}\label{remark::irreducible_claim}
All ANI instances reported in the literature are original-irreducible.
\end{remark}

\section{Polynomial Algorithms for the AI Class}~\label{sec::AI}
In this section, we present polynomial-time algorithms devoted to the solution of the generalized AI class. We first present a naive algorithm, and then follow with two improved variants that make use of some properties that characterize the benchmark instances.

We have a few observations before describing the algorithms.

\begin{observation}
In the BPP, items with identical weights are interchangeable. Therefore, when identifying an item, an algorithm only needs to determine the corresponding weight rather than the specific item index. Thus, in the following, when no confusion arises, we may simply use $i$ to denote an item of weight $w_i$.
\end{observation}

\begin{observation}
The algorithms in this section are described for the BPP, but they can also be applied to CSP and SSP instances with arbitrary demands after reducing them to an equivalent unit-demand instance \(\widehat I\). Given a CSP/SSP instance \(I\) and constants \(\alpha,\beta,\gamma \in \mathbb{Z}_+\), the corresponding unit-demand instance \(\widehat I\) can be a \((\alpha,\beta,\gamma)\)-AI instance only if
\begin{equation}~\label{eq::bpp_reduction}
	\sum_{i\in I} d_i \le \gamma |I|+\alpha+1,
\end{equation}
since each \(\gamma\)-tuple must contain at least one item with a distinct weight in \(I\), and in the worst case every item in the original instance may have a distinct weight. Condition~\eqref{eq::bpp_reduction} can be checked in \(O(|I|)\) time, and the reduction to \(\widehat I\) needs to be performed only when~\eqref{eq::bpp_reduction} holds.
\end{observation}

\begin{observation}
All algorithms in this section can be executed on any BPP instance. However, they are guaranteed to produce a perfect packing solution (and hence an optimal solution) only when the instance belongs to the generalized AI class.
\end{observation}

\begin{observation}
With minor adjustments, algorithms for the AI class can also be adapted to find an optimal solution for the generalized ANI instances. The main distinction lies in the computation of a tight lower bound. For AI instances, the lower bound $\lceil \sum_{i \in I} w_i / W \rceil$ is tight by construction, so the challenge is just to build an optimal solution. On the other hand, we are not aware of any polynomial-time procedure for computing a tight lower bound for ANI instances or for determining whether a given instance is ANI. (This is further discussed in Section~\ref{sec::ANI}, where we present pseudopolynomial-time algorithms for the ANI instances.)
\end{observation}

\subsection{A Naive Polynomial-Time Algorithm for the AI Class}
A naive polynomial-time algorithm for solving \textsc{AI} instances proceeds as follows. Let $\mathcal{I}'$ be an AI instance, with $O'$ denoting its corresponding original instance. Since $O'$ consists of $16$ items, we enumerate all $\binom{n}{16}$ candidate subsets $C \subseteq \mathcal{I}'$ of size $16$. As $|C|$ is constant, checking whether $C$ has a perfect packing solution with $3$ bins can be done in constant time, as explained in Remark~\ref{remark:trivial-algorithm}.

\begin{remark}\label{remark:trivial-algorithm}
A BPP instance $I$ with a constant number $n$ of items can be solved in constant time by a simple backtracking algorithm. Starting with $K=\lceil \sum_{i \in I} w_i / W \rceil$, we test whether $I$ admits a solution using $K$ bins by recursively assigning each item only to bins with sufficient residual capacity. If no such solution exists, $K$ is increased by one, and the procedure is repeated. The first solution found, that is, the first packing containing all items, is optimal. For a fixed $K$, at most $K^n$ assignments are considered. Hence, the running time is constant whenever $n$ is constant.
\end{remark}

For each candidate set $C \subseteq \mathcal{I}'$ having a perfect packing solution with $3$ bins, we attempt to find an optimal solution to the residual instance $J = \mathcal{I}' \setminus C$ using the following Observation~\ref{ob:residual} and Theorem~\ref{th:unique_triplet}. The latter relies on the concept of a \emph{full weight triplet}, defined as follows.

\begin{definition}
A \emph{full weight triplet} is a triplet of weights $(w_\ell,w_p,w_q)$ whose sum equals $W$. Given a weight triplet $(w_\ell,w_p,w_q)$, the corresponding item triplet $(\ell,p,q)$ may be any set of distinct items in $\mathcal{I}'$ with weights $w_\ell$, $w_p$, and $w_q$, respectively.
\end{definition}

We say that an item $\ell \in \mathcal{I}'$ belongs to a \emph{unique} full weight triplet $(w_\ell,w_p,w_q)$ if $w_p$ and $w_q$ form the only combination of weights such that $(\ell,p,q)$ is a full triplet in $\mathcal{I}'$. Note that uniqueness is defined only with respect to weights: even if multiple items in $\mathcal{I}'$ share the same weight as $p$ or $q$, the triplet is still considered unique provided no other weight combination is possible.

\begin{observation}~\label{ob:residual}
	Given an original triplet $(\ell_k, p_k, q_k)$, the instance \(J = \mathcal{I}' \setminus \{\ell_k, p_k, q_k\}\) is an AI instance with \(O'\) as its original instance.
\end{observation}

\begin{theorem}\label{th:unique_triplet}
	Suppose \(C = O'\) and $J \neq \emptyset$. Then:
	\begin{enumerate}[(i)]
		\item There exists at least one item \(\ell \in J\) that belongs to exactly one full weight triplet \((w_\ell, w_p, w_q)\);
		\item For every item \(\ell\) satisfying (i), its full weight triplet is equal to \((w_{\ell_k},w_{p_k},w_{q_k})\) for some \(k \in \{1,\ldots,h\}\).
	\end{enumerate}
\end{theorem}

\begin{proof}
	To prove (i), note that since \(C = O'\), the original triplet \((\ell_h, p_h, q_h)\) is contained in \(J\), so \(\ell_h \in J\) belongs to at least one full triplet. Additionally, condition~\eqref{cond:no-completion} implies that every full triplet containing \(\ell_h\) also contains \(p_h\) or \(q_h\) (or both). Given a full triplet containing both \(\ell_h\) and \(p_h\) (the argument is analogous for \(\ell_h\) and \(q_h\)), the third item must have weight \(w_{q_h}\). Hence, every full triplet containing \(\ell_h\) has the form \((w_{\ell_h}, w_{p_h}, w_{q_h})\).

	For (ii), suppose, by contradiction, that there exists an item \(\ell \in J\) and a full weight triplet \((w_\ell, w_p, w_q)\) that is not equal to \((w_{\ell_k}, w_{p_k}, w_{q_k})\) for any \(k = 1, \ldots, h\). Since every \(\ell \in J\) belongs to at least one original triplet, this implies that \(\ell\) belongs to at least two distinct full weight triplets, contradicting (i). Therefore, an item \(\ell\) satisfying (i) also satisfies (ii).
\end{proof}

With these results, for each candidate set \(C\) that fits exactly into three bins, we attempt to recover the \(h\) triplets from \(J = \mathcal{I}' \setminus C\) as follows. If there is an item \(\ell \in J\) satisfying condition (i), we fix its corresponding triplet $(\ell, p, q)$ in a partial solution $\mathbb{S}$ and apply the argument recursively to the residual instance $J\setminus \{\ell, p, q\}$. If $C = O'$, then only original triplets are removed, and the item $\ell_k$ with the largest index \(k\) in the residual instance always satisfies condition (i), which makes $\mathbb{S}$ correspond to an optimal solution to $J$ after $h$ iterations. Concatenating this solution with the one for \(C\) produces an optimal solution for \(\mathcal{I}'\). If at some iteration, for a given candidate $C$, the residual instance is nonempty and no item satisfies condition (i), then the candidate set $C$ is not equal to \(O'\), and we proceed to the next one.


Algorithm~\ref{alg::naive_AI} formalizes this process. Lines~2--4 run in \(O(n)\) time, while lines~5--6 run in \(O(n^2)\) time, since the number of full weight triplets is bounded by \(O(n^2)\); once the first two weights are fixed, the third is uniquely determined. Each iteration of the \texttt{while} loop runs in $O(n)$, and the maximum number of iterations is $h \in O(n)$. Therefore, the overall running time of the algorithm is $O(n^{18})$. Although this is polynomial, it is clearly impractical for instances reported in the literature, where typically $n \geq 202$.

\begin{algorithm}[!htbp]
	\setstretch{1}
	\caption{Naive Polynomial-Time Algorithm for an AI Instance}
	\label{alg::naive_AI}
	\DontPrintSemicolon
	\KwIn{An arbitrary instance $I$}
	\KwOut{An optimal solution for $I$, guaranteed when $I$ is an \textsc{AI} instance, or \textsc{Fail}}
	\ForEach{subset $C \subseteq I$ with $|C| = 16$}{
		\textbf{if} $C$ does not have a perfect packing solution with $3$ bins \textbf{then} continue\;
		$J \gets I \setminus C$\;
		$\mathbb{S} \gets$ an optimal solution for $C$ \tcp*{Partial solution for $I$}
		Compute all full weight triplets in $J$\;
		For each $i \in J$, let $\tau(i)$ be the number of full weight triplets containing $i$\;
		\While{$J \neq \emptyset$}{
			Let $i^* \in \arg\min_{i \in J} \tau(i)$\;
			\textbf{if} $\tau(i^*) \neq 1$ \textbf{then} \textbf{break} \tcp*{Wrong choice of $C$}
			Let $(w_\ell, w_p, w_q)$ be the unique full weight triplet containing $i^*$\;
			$\mathbb{S} \gets \mathbb{S} \cup \{(\ell, p, q)\}$\;
			$J \gets J \setminus \{\ell, p, q\}$\;
			Update $\tau(i)$ for all $i \in J$\;
		}
		\textbf{if} $J = \emptyset$ \textbf{then} \Return $\mathbb{S}$ \tcp*{Optimal solution found}
	}
	\Return Fail \tcp*{It never happens in an AI instance}
\end{algorithm}

For the \((\alpha,\beta,\gamma)\)-\textsc{AI} instances, since all \(\beta\)-tuples with total weight \(W\) can be enumerated in time \(O(n^{\beta - 1})\), and \(|\widetilde{O}'| = \alpha + 1\), which allows all candidate sets to be enumerated in time \(O(n^{\alpha + 1})\), an analogous algorithm can be obtained with overall time complexity \(O(n^{\alpha + \beta})\).

\FloatBarrier
\subsection{A Second Polynomial-Time AI Algorithm}~\label{sec::AI_2}

For the case in which condition~\eqref{cond:a_is_large} is satisfied, which we recall is verified for all benchmark instances in the literature, we can design a more efficient algorithm. Let \(L\) denote the set of \emph{large} items \(i \in \mathcal{I}'\) with \(w_i \geq W/2\). We identify all items \(\ell_k\) by exploiting the observation that \(\{\ell_k\colon k \in \{1, \ldots, h\}\} \subseteq L\). By Remark~\ref{remark:distinct_weights}, we can define \(L\) so that it contains at most one item per distinct weight. Since the original instance \(O'\) satisfies \(z_{\mathrm{ILP}} = 3\), we have that \(|L| \leq h + 3\).

Let \(\mathcal{T}\) be the set of all full weight triplets, and let \(\tau(i)\), for each $i \in \mathcal{I}'$, denote the number of full weight triplets in \(\mathcal{T}\) in which item \(i \in \mathcal{I}'\) appears. For $k = 1, \ldots, n^2$, let $B_k \subseteq L$ be the subset of large items $i$ such that $\tau(i) = k$. Additionally, we observe that the original instances selected from~\citet{Caprara_2015} contain either one or two large items, but never three. Consequently, the \textsc{AI} and \textsc{ANI} instances considered in the literature satisfy \(|L| \in \{h+1,h+2\}\), with one large item coming from each original triplet and the remaining one or two large items coming from the original instance.

Observe that, if the sets \(B_k\) are computed with respect to the instance \(\mathcal{I} = (\mathcal{I}' \setminus O') \cup O\), that is, the \textsc{ANI} instance from which \(\mathcal{I}'\) is derived, then \(\ell_h \in B_1\) by the definition of the instance. However, this property does not necessarily hold for instance \(\mathcal{I}'\), since items \(o_1\) and \(o_2\) may introduce additional triplets involving \(\ell_h\). We address this issue by considering a residual instance \(J = \mathcal{I}' \setminus \{s_1, s_2\}\) and the corresponding set \(L' = L \setminus \{s_1, s_2\}\), where \(s_1\) and \(s_2\) are two arbitrary items. There are \(\binom{n}{2}\) possible choices for pair \((s_1, s_2)\), and if \((s_1, s_2) = (o_1, o_2)\), then \(\ell_h \in B_1\) in the instance \(J\).

An item $\ell \in B_1$ is called a \emph{triplet candidate}, since there is a unique way to pack $\ell$ into a full weight triplet. Moreover, if $\ell = \ell_k$ for some $k \in \{1,\ldots,h\}$, we call it a \emph{true triplet candidate}. However, an item $\ell \in B_1$ may belong to $O'$, in which case we do not know whether $\ell$ must be packed into a triplet. Such items, that is, items $\ell \in B_1 \cap O'$, are therefore called \emph{false triplet candidates}. The number of false triplet candidates is $|L'|-h$, and we can enumerate all subsets $F \subseteq L'$ of size $|L'|-h$ in $O(n^{|L'|-h})$ time, which is $O(n^3)$ in the worst case.

For at least one such subset \(F\), all triplet candidates in the residual instance \(J \setminus F\) are true triplet candidates. For this residual instance, we select an item \(\ell \in B_1\), fix its corresponding triplet \(T\) in a partial solution, remove the items of \(T\), and repeat this procedure. After \(h\) iterations, all original triplets are recovered, and the remaining items together with \(\{s_1, s_2\} \cup F\) form the original instance \(O'\).

Observe that there are $O(n^5)$ possible choices for the set $\{s_1,s_2\} \cup F$, all of which are enumerated by this algorithm. If we initially compute the set of all full weight triplets \(\mathcal{T}\) and the sets \(B_k\), then, for a given candidate \(\{s_1, s_2\} \cup F\), we can obtain the set of valid triplets \(\mathcal{T}'\) for the instance \(J = \mathcal{I}' \setminus (\{s_1, s_2\} \cup F)\) in time complexity \(O(n^2)\) by removing invalid triplets from \(\mathcal{T}\) and updating the corresponding sets \(B_k\).

If there exists an item \(\ell \in B_1\), we take the unique triplet \(T\) containing \(\ell\), fix it in the partial solution, and remove the items from $T$. If \(B_1 = \emptyset\), then the candidate set \(\{s_1, s_2\} \cup F\) is incorrect, and we abort the current iteration and proceed to the next candidate set.

The algorithm is deemed successful if it recovers \(h\) triplets and the remaining items, that is, the items not packed into these \(h\) triplets, can be packed into three bins. For each candidate set \(\{s_1, s_2\} \cup F\), at most \(n\) iterations are performed, each with amortized time complexity \(O(n)\). To see this, recall that \(O(n^2)\) triplets are considered initially. In each iteration, we identify an item \(\ell\) that appears in the fewest triplets, which requires \(O(n)\) time, and then fix its associated triplet \((\ell,p,q)\) by removing all invalid triplets. Over all iterations, this invalidation step takes \(O(n^2)\) time in total. Consequently, the overall time complexity of the algorithm is \(O(n^7)\). This bound can be expressed as a function of \(|L|-h\), yielding a time complexity of \(O(n^{4+|L|-h})\). Since \(|L|-h \in \{1,2\}\) for the literature instances, the time complexity reduces to \(O(n^5)\) when \(|L|-h = 1\) and to \(O(n^6)\) when \(|L|-h = 2\), representing a significant improvement over Algorithm~\ref{alg::naive_AI}.

Finally, for the generalized \textsc{AI} instances that satisfy condition~\eqref{cond:a_is_large}, the algorithm is sensitive to both the size of \(F\) and the number of items \(\gamma\) per tuple. Given an \((\alpha, \beta, \gamma)\)-\textsc{AI} instance and noting that \(|F| \leq |L| - h \leq \beta\), there are \(O(n^{\beta + 2})\) candidate sets. For each candidate set, the algorithm enumerates all full patterns consisting of \(\gamma\)-tuples, whose number is bounded by \(O(n^{\gamma - 1})\). Since \(\gamma \geq 3\), this enumeration dominates the running time for each candidate. Therefore, the overall time complexity is \(O(n^{\beta + \gamma + 1})\).

\subsection{A Practical Polynomial-Time AI Algorithm}
\label{sec::AI_3}
The second \textsc{AI} algorithm performs significantly better than the first when condition~\eqref{cond:a_is_large} holds, but its runtime is still not practical. Fortunately, candidate sets must satisfy additional structural properties to be considered promising. Exploiting these properties allows us to devise an improved algorithm that, in practice, examines only a very small number of candidate items, as discussed below.

The first observation is that if $B_1 \neq \emptyset$ for the current instance $J$, then it is unnecessary to immediately select both $s_1$ and $s_2$. More precisely, one can select an item $\ell \in B_1$ and its corresponding triplet $T$, and then decide either to add $T$ to the partial solution or to remove $\ell$ (assuming $\ell \in O'$). The latter option corresponds to adding $\ell$ to $F$ and can occur at most $|L| - h$ times.

Conversely, if \(B_1 = \emptyset\), it is not necessary to consider all items as candidates for \(s_1\). Instead, we restrict attention to the set of items $\mathcal{S}$ such that for each $s \in \mathcal{S}$ there exists an item \(i \in B_2 \cup B_3\) such that \(i\) and \(s\) belong to the same full weight triplet. Observe that both \(s_1\) and \(s_2\) must belong to $\mathcal{S}$, since $\mathcal{S}$ consists exactly of the items whose removal results in a non-empty set \(B_1\). A similar procedure is applied to select $s_2$ when $B_1 = \emptyset$ but $s_1$ has already been chosen; in this case, the selection is restricted to items from patterns in $B_2$.

Algorithm~\ref{alg:final-AI} illustrates how these ideas are implemented in practice. The algorithm is a recursive backtracking procedure with a polynomially bounded number of recursive calls. It receives as input the residual instance $J$ (initially $J = \mathcal{I}'$), the set $\mathcal{T}$ of full weight triplets, the function $\tau$ maps each item $i$ to the number of triplets in which it appears, the family $\mathcal{B}$ of sets $B_k = \{ i \in J\colon \tau(i) = k \}$ for $k \in \{1, \ldots, n^2\}$, the bin capacity $W$, a counter $n_s$ indicating how many split items ($s_1$ or $s_2$) have already been chosen, a counter $n_b$ indicating how many large items cannot be packed using triplets, the set $R$ of removed items corresponding to the original instance $O'$, and the partial solution $\mathbb{S}$. It uses the constant $n_b^{\max}$ defined as $|L| - h$, which corresponds to the number of large items in $L$ that belong to the original instance.

\begin{algorithm}[!t]
	\setstretch{1}
	\caption{Practical AI Solver Algorithm}
	\label{alg:final-AI}
	\DontPrintSemicolon
	\SetKwInOut{Input}{Input}
	\SetKwInOut{Output}{Output}
	\SetKwProg{Fn}{Function}{:}{}
	\Fn{$\mathrm{PracticalAISolver}(J, \mathcal{T}, \tau, \mathcal{B}, W, n_s, n_b, R, \mathbb{S})$}{
	\Input{Residual instance $J$, triplets $\mathcal{T}$, function $\tau$,\\
        sets $\mathcal{B} = \{B_1, \ldots, B_{n^2}\}$, capacity $W$,\\
        split count $n_s$, large count $n_b$, removed items $R$,\\
        partial solution $\mathbb{S}$}
	\Output{An optimal solution for the first-call instance, or \textsc{Fail}}
		\If{$J = \emptyset$}{
			\tcp{All remaining items are in $R$}
			\textbf{if} there is a perfect packing solution $\mathbb{S}_R$ for $R$ using 3 bins \textbf{then} \Return $\mathbb{S} \cup \mathbb{S}_R$\;
			\textbf{else} \Return \textsc{Fail}\;
		}
		\textbf{if} $\mathcal{T} = \emptyset$  \textbf{or} $n_b > n_b^{\max}$  \textbf{or} $|R| > 16$ \textbf{or} $\sum_{i \in R} w_i > 3W$ \textbf{then}
			\Return \textsc{Fail}\;
		\eIf{$B_1 = \emptyset$}{
			\textbf{if} $n_s = 2$ \textbf{then} \Return \textsc{Fail}\;
			$B \gets B_2$\;
			\textbf{if} $n_s = 0$ \textbf{then} $B \gets B \cup B_3$\;
			\ForEach{$i \in B$}{
				$S_i \gets \{ s \in J : \exists T \in \mathcal{T} \text{ with } \{i,s\} \subseteq T \}$\;
			}
			$\mathcal{S} \gets \bigcup_{i \in B} S_i$\;
			\ForEach{$s \in \mathcal{S}$}{
				$(\widehat{J}, \widehat{\mathcal{T}}, \widehat{\tau}, \widehat{\mathcal{B}}, \widehat{R}, \widehat{n}_b) \gets \mathrm{RemoveItem}(J, \mathcal{T}, \tau, \mathcal{B}, R, n_b, s)$\;
				$sol \gets \mathrm{PracticalAISolver}(\widehat{J}, \widehat{\mathcal{T}}, \widehat{\tau}, \widehat{\mathcal{B}}, W, n_s+1, \widehat{n}_b, \widehat{R}, \mathbb{S})$\;
				\textbf{if} $sol \neq \textsc{Fail}$ \textbf{then} \Return $sol$\;
			}
			\Return \textsc{Fail}\;
		}{
			Select $\ell \in B_1$ and let $T = (w_\ell,w_p,w_q)$ be the triplet containing $\ell$\;
			$(\widehat{J}, \widehat{\mathcal{T}}, \widehat{\tau}, \widehat{\mathcal{B}}, \widehat{R}, \widehat{n}_b) \gets \mathrm{RemoveItems}(J, \mathcal{T}, \tau, \mathcal{B}, R, n_b, \{\ell, p, q\})$\;
			$sol \gets \mathrm{PracticalAISolver}(\widehat{J}, \widehat{\mathcal{T}}, \widehat{\tau}, \widehat{\mathcal{B}}, W, n_s, \widehat{n}_b, \widehat{R}, \mathbb{S} \cup \{(\ell, p, q)\})$\;
			\textbf{if} $sol \neq \textsc{Fail}$ \textbf{then} \Return $sol$\;
			$(\widehat{J}, \widehat{\mathcal{T}}, \widehat{\tau}, \widehat{\mathcal{B}}, \widehat{R}, \widehat{n}_b) \gets \mathrm{RemoveItem}(J, \mathcal{T}, \tau, \mathcal{B}, R, n_b, \ell)$\;
			\Return $\mathrm{PracticalAISolver}(\widehat{J}, \widehat{\mathcal{T}}, \widehat{\tau}, \widehat{\mathcal{B}}, W, n_s, \widehat{n}_b, \widehat{R}, \mathbb{S})$\;
		}
	}
\end{algorithm}

Line~2 checks whether the residual instance is empty, which indicates that $h$ triplets have been found. In this case, it remains to verify whether $R$ fits into three bins. If so, an optimal solution for the complete instance is obtained, and failure is reported otherwise. Line~5 tests the existence of triplets that pack $J$ and verifies whether $R$ is consistent with a valid original instance.

If $B_1 = \emptyset$, we enter the condition at line~6 and must choose an item to be removed to make $B_1$ nonempty, which is possible only if $n_s < 2$. The set $\mathcal{S}$ contains all items that may correspond to $s_1$ or $s_2$. The function \textsc{RemoveItem} recomputes the input arguments after removing an item $s \in \mathcal{S}$. This operation may remove additional items, since $R$ contains all items that cannot or should not be packed into triplets; after removing $s$, some items may no longer belong to any triplet and must therefore be added to $R$. A recursive call is then performed, which either fails or returns an optimal solution.

Line~19 handles the case where there exists an item $\ell \in B_1$. In this case, we select the triplet $T = (w_\ell,w_p,w_q)$ and add it to the partial solution $\mathbb{S}$, then remove its items and recompute the input arguments for the recursive call. Note that removing $(\ell,p,q)$ may force the removal of additional large items, since $p$ or $q$ may also belong to their unique triplet. For this reason, $\widehat{R}$ and $\widehat{n}_b$ must be recomputed. Line~21 performs the recursive call. If it fails, this may indicate that $\ell$ is a false triplet candidate; in that case, $\ell$ is removed and another recursive call is executed. If $n_b = n_b^{\max}$ in the current call, then the recursive call immediately reports failure at the base case.

In practice, the algorithm examines only a small number of candidates, empirically bounded by $k n$ for some small constant $k$. This represents a substantial reduction from the $O(n^5)$ possibilities of the naive approach. For a given candidate, if the function arguments are passed by reference and the arguments for the next call are computed from the current ones while keeping track of modifications to be undone after the call returns, the total time spent along any root-to-leaf path (i.e., from the initial call to a base case where $h$ triplets are recovered) is $O(n^2)$, as in the previous algorithm. Therefore, in practice, the algorithm exhibits a running time close to $O(n^3)$.

\section{Pseudopolynomial Algorithms for the ANI Class}~\label{sec::ANI}
As shown in the previous section, instances of the generalized \textsc{AI} class can be solved in polynomial time. This result relies on the fact that such instances admit a perfect packing solution, whose optimality is straightforward to prove. Obtaining an analogous result for the generalized \textsc{ANI} class is substantially more challenging. Although an algorithm similar to the one used for the \textsc{AI} case can be applied to find a solution using $z_{\mathrm{LP}}^{\widetilde{\mathcal{I}}} + 1$ bins for a generalized \textsc{ANI} instance $\widetilde{\mathcal{I}}$, proving the optimality of this solution requires showing that no perfect packing exists. This additional requirement considerably increases the difficulty of deriving a polynomial-time algorithm for this class. Nevertheless, we derive a pseudopolynomial-time algorithm for $(\alpha,\beta,3)$-\textsc{ANI} instances and show that, when combined with additional techniques, it is significantly faster than state-of-the-art methods on the \textsc{ANI} instances from the literature.

Similar to the \textsc{AI} case, we first present a theoretical pseudopolynomial algorithm for the \textsc{ANI} class. We then show how additional information can be used to obtain faster algorithms in practice. The algorithms are presented for the BPP, but they can be easily adapted for the SSP, as explained at the end of the section.

\subsection{A Pseudopolynomial Time Algorithm for the \texorpdfstring{$(\alpha, \beta, 3)$}{(alpha, beta, 3)}-ANI Classes}
For an \((\alpha,\beta,3)\)-\textsc{ANI} instance \(\widetilde{\mathcal{I}}\), we have
\[
	\sum_{i \in \widetilde{\mathcal{I}}} \frac{w_i}{W} = z_{\mathrm{LP}}^{\widetilde{\mathcal{I}}} \in \mathbb{Z}_+,
\]
which may initially suggest the existence of a solution using \(z_{\mathrm{LP}}^{\widetilde{\mathcal{I}}}\) bins, that is, a perfect packing. However, in contrast to the \textsc{AI} class, it is not sufficient to simply identify a candidate set \(C \subseteq \widetilde{\mathcal{I}}\) corresponding to the original instance $\widetilde O$, since this does not rule out the existence of perfect packing solutions that use patterns combining items from \(\widetilde O\) and \(\widetilde{\mathcal{I}} \setminus \widetilde O\).

For this reason, we address the \textsc{ANI} class by applying a sequence of reductions in which patterns that must appear in at least one perfect packing solution, if such a solution exists, are fixed. At the end of this process, we obtain a reduced instance of constant size for which it is straightforward to verify that no perfect packing solution exists. In addition, a solution for the reduced instance using one additional bin can be found in constant time (Remark~\ref{remark:trivial-algorithm}), which yields a solution for \(\widetilde{\mathcal{I}}\) using \(z_{\mathrm{LP}}^{\widetilde{\mathcal{I}}} + 1\) bins and is therefore optimal.

While the algorithms for \textsc{AI} instances extend naturally to generalized \textsc{AI} instances, our algorithm for the \textsc{ANI} class is restricted to $(\alpha,\beta,3)$-\textsc{ANI} instances, and this restriction is fundamental. Conceptually, our approach relies on the fact that if we identify a pair of items that appear together in a perfect packing solution, whenever such a solution exists, then the third item is uniquely determined. In contrast, for $\gamma$-tuples with $\gamma > 3$, this structural property no longer holds, and it is unclear whether an analogous algorithm can be derived.

\begin{theorem}~\label{theorem::partition_triplet} Given an original-irreducible \((\alpha,\beta,3)\)-ANI instance \(\widetilde{\mathcal{I}}\) (see definition~\ref{definition::irreducible}). Let \(\ell,p,q \in \widetilde{\mathcal{I}}\) be pairwise distinct items such that
	(i) \((\ell,p,q)\) is a full triplet, and
	(ii) every full pattern in \(\widetilde{\mathcal{I}}\) containing \(\ell\) also contains \(p\) or \(q\) (or both).
	Then there exist pairwise distinct items \(\ell', p', q' \in \widetilde{\mathcal{I}}\) such that \(w_{\ell'} = w_\ell\), \(w_{p'} = w_p\), \(w_{q'} = w_q\), and either \(\{\ell', p', q'\} \subseteq \widetilde O\) or \(\{\ell', p', q'\} \subseteq \widetilde{\mathcal{I}} \setminus \widetilde O\).
\end{theorem}

\begin{proof}
Let \(\ell,p,q \in \widetilde{\mathcal{I}}\) satisfy the hypotheses (i) and (ii). We have two cases.

\smallskip
\noindent Case 1: $\ell \in \widetilde{\mathcal{I}} \setminus \widetilde O$.
Then \(w_\ell\) must coincide with one of the weights in an original triplet, that is,
\(w_\ell \in \{w_{\ell_k}, w_{p_k}, w_{q_k}\}\)
for some \(k \in \{1,\ldots,h\}\).
Without loss of generality, assume \(w_\ell = w_{\ell_k}\).
If \(w_p \in \{w_{p_k}, w_{q_k}\}\) or \(w_q \in \{w_{p_k}, w_{q_k}\}\), then, as $(\ell, p, q)$ is a full triplet,
$\{w_p, w_q\} = \{w_{p_k}, w_{q_k}\}$, and triplet
\((\ell_k,p_k,q_k)\) satisfies the claim.
Otherwise, \(p_k,q_k \in \widetilde{\mathcal{I}} \setminus \{\ell,p,q\}\) and,
in this case, \((\ell,p_k,q_k)\) forms a full pattern containing \(\ell\) but neither \(p\) nor \(q\), contradicting assumption~(ii).

\smallskip
\noindent Case 2: $\ell \in \widetilde O$.
If \(p,q \in \widetilde O\), the result follows. Otherwise, assume \(q \in \widetilde{\mathcal{I}} \setminus \widetilde O\).
If the same were true for \(p\) (i.e., \(p \in \widetilde{\mathcal{I}} \setminus \widetilde O\)), then only \(\ell \in \widetilde O\),
and since every item in \(\widetilde O\) must belong to at least one full pattern \(P \subseteq \widetilde O\),
such a pattern containing \(\ell\) would contain neither \(p\) nor \(q\), contradicting assumption~(ii). Therefore, \(p \in \widetilde O\) and \(q \in \widetilde{\mathcal{I}} \setminus \widetilde O\). Moreover, if there exists a full pattern \(P \subseteq \widetilde O\) containing \(\ell\) but not \(p\), then $P$ contains neither $p$ nor $q$ and thus contradicts assumption~(ii). Hence, every full pattern \(P \subseteq \widetilde O\)
containing \(\ell\) also contains \(p\), which is a contradiction since $\widetilde{\mathcal{I}}$ is an original-irreducible instance. Consequently, the only feasible possibility is $\ell,p,q \in \widetilde O$.
\end{proof}

Lemma~\ref{lemma::full_pair} and Theorem~\ref{theorem::mandatory_triplet} show that if a perfect packing solution exists for an instance $I$ of an arbitrary class, then a triplet $(\ell, p, q)$ satisfying the hypothesis of Theorem~\ref{theorem::partition_triplet} must belong to at least one perfect packing solution.

\begin{lemma}~\label{lemma::full_pair}
	Given an arbitrary BPP instance $I$, suppose that there are two items $i$ and $j$ such that $w_i + w_j = W$. Then, there exists an optimal solution in which $i$ and $j$ are packed together. The result also extends to the CSP and the SSP. Given two item types $i$ and $j$ satisfying the same condition, possibly with $i=j$, there exists an optimal solution in which the pattern $(i,j)$ is used $\min\{d_i,d_j\}$ times if $i \neq j$, and $\lfloor d_i/2 \rfloor$ times if $i=j$.
\end{lemma}

\begin{proof}
	We first consider a BPP instance $I$. Let $\mathbb{S}$ be a feasible (but not necessarily optimal) solution for $I$ using $K$ bins. If $i$ and $j$ are not packed together, let $S$ denote the (possibly empty) set of items packed together with $i$ in solution $\mathbb{S}$. Note that the total weight $w(S)$ of $S$ is at most $w_j$. Since $j$ is packed in a different bin, we can swap $j$ with $S$ to obtain another feasible solution $\mathbb{S}'$ using at most $K$ bins (the bin containing $j$ in $\mathbb{S}$ may become empty if $S = \emptyset$). If $\mathbb{S}$ is optimal, then $\mathbb{S}' \setminus \{(i,j)\}$ is an optimal solution for $I \setminus \{i, j\}$.

	For the CSP case, analogous arguments apply to copies of $i$ and $j$ that are not packed together. For the SSP, the argument is similar to the CSP case, with $S$ now denoting a set of items whose total weight is at least $w_j$.
\end{proof}
\begin{theorem}\label{theorem::mandatory_triplet}
	Let $I$ be an arbitrary instance. Additionally, let $\ell,p,q \in I$ be pairwise distinct items such that $(\ell,p,q)$ is a full triplet, and every full pattern in $I$ containing $\ell$ also contains $p$ or $q$ (or both). If a perfect packing solution exists for $I$, then there exists a perfect packing solution for $I$ that includes the triplet $(\ell,p,q)$.
\end{theorem}

\begin{proof}
	By hypothesis, in every perfect packing solution, the full pattern containing $\ell$ also contains $p$ or $q$ (or both). Assume that a perfect packing solution exists in which $\ell$ is packed with $p$, and consider the instance $\widehat{I}$ obtained by replacing $\ell$ and $p$ with a single item $c$ of weight $w_c = w_\ell + w_p$. Clearly, $\widehat{I}$ also admits a perfect packing solution. By Lemma~\ref{lemma::full_pair}, there exists a perfect packing solution for $\widehat{I}$ that includes both $c$ and $q$, which corresponds to a perfect packing solution for $I$ containing the pattern $(\ell,p,q)$. A symmetric argument holds if $\ell$ is packed with $q$. Therefore, if a perfect packing solution exists for $I$, there exists one that uses the pattern $(\ell,p,q)$.
\end{proof}

Lemmas~\ref{lemma::mandatory_triplet_lp} and ~\ref{lemma::reduction_ANI} establish the correctness of the reduction utilized by our algorithms.

\begin{lemma}\label{lemma::mandatory_triplet_lp}
	Let $I$ be an arbitrary instance with bin capacity $W$, and let $\ell,p,q \in I$ be pairwise distinct items such that $(\ell,p,q)$ is a full triplet, and every full pattern in $I$ containing $\ell$ also contains $p$ or $q$ (or both). If there exists an optimal SCF LP solution $x$, that is, an optimal solution to the LP relaxation of~\eqref{scf:obj}--\eqref{scf:int}, with value $z(x)=\sum_{i\in I} w_i/W \in \mathbb{Z}_+$, then there exists an optimal LP solution $x'$ with $x'_{(\ell,p,q)}=1$.
\end{lemma}

\begin{proof}
	By $z(x)=\sum_{i\in I} w_i/W \in \mathbb{Z}_+$, $x$ can use only full patterns. Define
    \begin{itemize}
        \item $\mathcal{P}^{\overline{\ell}, i} = \{P \in \mathcal{P}\colon x_P > 0,\ \ell \notin P,\ i \in P\}$ for $i \in \{p, q\}$,
        \item $\mathcal{P}^{\ell} = \{P \in \mathcal{P}\colon x_P > 0,\ \ell \in P\}$.
    \end{itemize}
    If $\mathcal{P}^{\overline{\ell},p}=\mathcal{P}^{\overline{\ell},q}=\emptyset$, then every positive pattern containing $p$ or $q$ also contains $\ell$, and hence $x_{(\ell,p,q)}=1$. Otherwise, assume without loss of generality that there exists a full pattern $Q \in \mathcal{P}^{\overline{\ell},q}$. Then there must also exist a full pattern $M \in \mathcal{P}^{\ell}$ with $q \notin M$; otherwise all patterns in $\mathcal{P}^{\ell}$ would contain both $\ell$ and $q$, implying either $\sum_{P\in\mathcal{P}^{\ell}} x_P<1$ or contradicting the existence of $Q$. By Theorem~\ref{theorem::mandatory_triplet}, every full pattern containing $\ell$ also contains $p$ or $q$, which implies $p \in M$.

	Let $S=M\setminus\{\ell,p\}$. Since $(\ell,p,q)$ is a full triplet and $M$ is full, $w(S)=W-w_\ell-w_p=w_q$, so $P'=(Q\cup S)\setminus\{q\}$ is a full pattern. Let $\delta=\min\{x_M,x_Q\}>0$ and replace $\delta$ units of $x_M$ and $x_Q$ by $\delta$ units of $x_{(\ell,p,q)}$ and $x_{P'}$. This preserves feasibility and optimality. Moreover, either $x_M=0$ or $x_Q=0$, and as $P'$ contains none of $\ell,p,q$, the quantity $|\mathcal{P}^{\overline{\ell}, p}|+|\mathcal{P}^{\overline{\ell}, q}|+|\mathcal{P}^{\ell}|$ strictly decreases, while $x_{(\ell,p,q)}$ strictly increases. Repeating this finitely many times yields an optimal SCF LP solution with $x_{(\ell,p,q)}=1$.
\end{proof}

\begin{lemma}\label{lemma::reduction_ANI}
	Given an \((\alpha,\beta,3)\)-ANI instance \(\widetilde{\mathcal{I}}\) and a triplet \((\ell,p,q)\) satisfying Theorem~\ref{theorem::partition_triplet}, the residual instance \(\widehat{\mathcal{I}} = \widetilde{\mathcal{I}} \setminus \{\ell,p,q\}\) is an \((\alpha',\beta',3)\)-ANI instance with $\alpha' \leq \alpha$ and $\beta' \leq \beta$.
\end{lemma}

\begin{proof}
	Recall that Theorem~\ref{theorem::partition_triplet} guarantees the existence of pairwise distinct items \(\ell',p',q'\), with the same respective weights as \(\ell,p,q\), such that all three items belong either to \(\widetilde O\) or to \(\widetilde{\mathcal{I}} \setminus \widetilde O\).
	If \(\ell', p', q' \in \widetilde{\mathcal{I}} \setminus \widetilde O\), then \(\widehat{\mathcal{I}}\) is an \((\alpha,\beta,3)\)-ANI instance with \(h-1\) triplets. If \(\ell', p', q' \in \widetilde O\), then \(\widehat{\mathcal{I}}\) can be classified as an \((\alpha-3,\beta-1,3)\)-ANI instance whose original instance is \(\widehat O = \widetilde O \setminus \{\ell',p',q'\}\). This follows from Lemma~\ref{lemma::mandatory_triplet_lp}, which guarantees that $x_{(\ell', p', q')} = 1$ in at least one optimal SCF relaxation solution \(x\) for \(\widetilde O\). By using Remark~\ref{remark:residual-gap}, we have that \(\widehat O\) is a Non-IRUP instance with gap~\(1\).
\end{proof}

Based on the lemmas above, our first algorithm consists of selecting a triplet $(\ell,p,q)$ satisfying Theorem~\ref{theorem::partition_triplet}, fixing it in a partial solution, and setting $\widehat{\mathcal{I}} = \widetilde{\mathcal{I}} \setminus \{\ell,p,q\}$. While $|\widehat{\mathcal{I}}| > \alpha$, the original triplet $(\ell_k,p_k,q_k)$ with the largest index $k$ that still belongs to $\widehat{\mathcal{I}}$ is one triplet that satisfies the theorem. After $h$ reductions, the residual instance $\widehat{\mathcal{I}}$ has size $\alpha$. Since $\alpha$ is constant, Remark~\ref{remark:trivial-algorithm} implies that we can determine in constant time that $\widehat{\mathcal{I}}$ admits no solution using $\beta$ bins and then find a solution using $\beta+1$ bins.

To identify triplets satisfying Theorem~\ref{theorem::partition_triplet}, we may apply a naive approach that tests, for each full weight triplet $(w_\ell,w_p,w_q)$, whether there exists an item $i \in \{\ell,p,q\}$ such that every full pattern containing $i$ also contains at least one of the other two items. A subset-sum DP over $\widehat{\mathcal{I}} \setminus \{\ell,p,q\}$ determines whether $i$ satisfies this condition in $O(nW)$ time. Since there are $O(n^2)$ full weight triplets, this procedure identifies a triplet $(\ell,p,q)$ satisfying Theorem~\ref{theorem::partition_triplet} in $O(n^3W)$ time. Consequently, recovering all triplets requires $O(n^4W)$ time, which is pseudopolynomial, but very slow in practice.

Moreover, for the application of this algorithm and those presented subsequently, it is not necessary to know in advance whether the instance is an original-irreducible $(\alpha,\beta,3)$-\textsc{ANI}. All algorithms apply to arbitrary instances while keeping the same running time. Nevertheless, the above results guarantee a reduction to constant size whenever the input is an original-irreducible $(\alpha,\beta,3)$-\textsc{ANI}.

Finally, although all \textsc{ANI} instances studied in the literature are original-irreducible, original-reducible $(\alpha,\beta,3)$-\textsc{ANI} instances can also be addressed. In this case, some triplets $(\ell,p,q)$ satisfying Theorem~\ref{theorem::partition_triplet} may contain a reducible pair $(\ell,p)$. This is handled by allowing the algorithm to forbid up to $\alpha - 1$ such triplets, each corresponding to the assumption that $\ell$ forms a reducible pair with either $p$ or $q$. The value $\alpha - 1$ represents the maximum number of merges possible before the original instance is reduced to a single item. The resulting algorithm branches on each triplet satisfying the theorem, deciding whether to accept or forbid it, yielding a branching tree of depth $O(n)$ with at most $\alpha-1$ forbidden decisions along any root-to-leaf path. Consequently, the tree has $O(n^\alpha)$ nodes, and at least one leaf corresponds to a residual instance of size $\alpha$. This adds a multiplicative factor of $n^\alpha$ to the running time, which remains pseudopolynomial.

\subsection{An Improved Pseudopolynomial Time Algorithm for Original-Irreducible \texorpdfstring{$(\alpha, \beta, 3)$}{(alpha, beta, 3)}-ANI Classes}~\label{sec::ANI_2}
Next, we present an improved algorithm for original-irreducible $(\alpha, \beta, 3)$-\textsc{ANI} instances that satisfy condition~\eqref{cond:a_is_large}, which is the case for literature instances. Without loss of generality, assume that
\[
	w_{\ell_k} > w_{p_k} \geq w_{q_k}, \qquad k \in \{1,\ldots,h\}.
\]
Given a triplet $(\ell,p,q)$ satisfying Theorem~\ref{theorem::partition_triplet}, any pattern containing item $\ell$ must have one of the following forms:
\[
	\{\ell, p, q^1, \ldots, q^\sigma\}
	\quad \text{or} \quad
	\{\ell, p^1, \ldots, p^\rho, q\},
\]
where $\{p^1,\ldots,p^\rho\} \subseteq \widetilde{\mathcal{I}} \setminus \{\ell,q\}$ and
$\{q^1,\ldots,q^\sigma\} \subseteq \widetilde{\mathcal{I}} \setminus \{\ell,p\}$ are sets of items whose total weights equal $w_p$ and $w_q$, respectively.

For the purposes of the algorithms in this section, an ordered tuple representation of a full pattern $P$ is obtained by arranging its items in non-increasing order of weight, with ties broken arbitrarily. A pattern may therefore admit multiple ordered representations when some items have equal weights. For example, consider a pattern $Q=\{\ell,p,q^1,\ldots,q^\sigma\}$, where $q^1,\ldots,q^\sigma$ are indexed in non-increasing order of weight. Then, $(\ell,p,q^1,\ldots,q^\sigma)$ is an ordered representation whenever $w_p\geq w_{q^1}$. If $\sigma>1$, then $w_{q^j}<w_p$ for every $j=1,\ldots,\sigma$, since $w_{q^j}<\sum_{k=1}^{\sigma}w_{q^k}=w_q\leq w_p$. Thus, every ordered representation of $Q$ has $p$ as its second item and $q^1,\ldots,q^\sigma$ as a suffix, possibly reordered among items of equal weight. If $\sigma=1$, $Q$ admits the representation $(\ell,p,q^1)$ and, when $w_p=w_{q^1}$, also the representation $(\ell,q^1,p)$.

Definition~\ref{def:mandatory_weight} and Lemma~\ref{lemma::mandatory_weight} highlight properties of what we call next \emph{mandatory weights}. Definition~\ref{definition::suffix_tuple} introduces the notion of a suffix of a tuple. Lemma~\ref{lemma::suff_property} provides a necessary condition for triplets $(\ell,p,q)$ satisfying Theorem~\ref{theorem::partition_triplet}, corresponding to an original triplet $(\ell_k, p_k, q_k)$.

\begin{definition}\label{def:mandatory_weight}
	A weight $w$ is \emph{mandatory} for an item $i$ if every full pattern containing $i$ includes either a distinct item of weight $w$ or a subset of items whose total weight is $w$.
\end{definition}

\begin{lemma} \label{lemma::mandatory_weight}
	Let $I$ be an arbitrary instance, and let $i,j \in I$ be such that $w_j$ is a mandatory weight for $i$. If $I$ has a perfect packing solution, then it has one in which $i$ and $j$ appear in the same pattern.
\end{lemma}

\begin{proof}
	If $i$ and $j$ are not packed together in a perfect packing solution, then the subset of items packed with $i$ whose total weight is $w_j$ can be switched with $j$. This yields another perfect packing in which $i$ and $j$ are together.
\end{proof}

\begin{definition}~\label{definition::suffix_tuple}
	A \emph{suffix} of a tuple $(t_1,t_2,\ldots,t_o)$ is any contiguous subsequence $(t_{o-s+1},\allowbreak\ldots,t_o)$ with $1 \le s \le o$.
\end{definition}

\begin{lemma}\label{lemma::suff_property}
	Given an original triplet $(\ell_k, p_k, q_k)$ satisfying Theorem~\ref{theorem::partition_triplet}, every pattern containing $\ell_k$ either contains an item of weight $w_{q_k}$ or has the property that all its ordered tuple representations admit a suffix whose total weight is $w_{q_k}$. This implies that $w_{q_k}$ is a \emph{mandatory weight} for $\ell_k$.
\end{lemma}

\begin{proof}
	For patterns of the form $\{\ell_k, p_k^1, \ldots, p_k^\rho, q_k\}$, the claim follows immediately, since an item of weight $w_{q_k}$ explicitly appears in the pattern. This is not necessarily the case for patterns of the form $\{\ell_k, p_k, q_k^1, \ldots, q_k^\sigma\}$. We claim that, in this case, every ordered tuple representation of such a pattern admits a suffix whose total weight equals $w_{q_k}$.

	If $p_k$ is the second item in the tuple, the claim follows immediately. The only case in which $p_k$ is not the second item is the tuple $(\ell_k, q_k^1, p_k)$, which occurs when $\sigma=1$ and $w_{p_k} = w_{q_k^1}$. This tuple also admits a suffix with total weight $w_{q_k}$, completing the proof.
\end{proof}

By Lemma~\ref{lemma::suff_property}, if we identify a triplet $(\ell,p,q) \in \widetilde{\mathcal{I}}$ satisfying the following conditions:
\begin{enumerate}
	\item $w_\ell \geq W/2$,
	\item $w_\ell + w_p + w_q = W$,
	\item $w_q$ is a mandatory weight for $\ell$,
\end{enumerate}
then $(\ell,p,q)$ is a candidate to satisfy Theorem~\ref{theorem::partition_triplet}. Such a triplet can be efficiently identified using dynamic programming~(DP), as described next.

Consider an $(\alpha, \beta, 3)$-\textsc{ANI} instance $\widetilde{\mathcal{I}}=\{1,\dots,n\}$ whose items are sorted in non-increasing order of weight. We define a DP table $\mathrm{dp}(i,w)$, for $1 \le i \le n+1$ and $0 \le w \le W$, to compute mandatory weights that appear as suffixes. For each item $i$ and partial capacity $w$, we compute the set of \emph{mandatory weights}, denoted as $\mathrm{dp}(i,w)$, required to go from partial capacity $w$ to the bin capacity $W$ when we are allowed to use only the items from $\{i,\ldots,n\}$. The base cases are defined as $\mathrm{dp}(n+1, W)=\emptyset$ and $\mathrm{dp}(n+1,w)=\bot$ for all $w<W$, where $\bot$ denotes an infeasible state, i.e., a state that cannot reach total capacity $W$ while satisfying the conditions.

For $i \in \{1,\ldots,n\}$ and $w \in \{0,\ldots,W\}$, the transition is computed as follows:
\[
	\mathrm{dp}(i,w)=
	\begin{cases}
		\mathrm{dp}(i+1,w),
		  & \text{if } w + w_i > W, \\[2pt]
		\{w_i\},
		  & \text{if } w + w_i = W, \\[2pt]
		\mathrm{dp}(i+1,w) \cap \bigl(\mathrm{dp}(i+1,w + w_i) \cup \{w_i\}\bigr),
		  & \text{otherwise}.
	\end{cases}
\]
In the above recurrence, the intersection of two infeasible states is defined as $\bot$, the intersection of an infeasible state $s$ and a feasible state $s'$ is equal to $s'$, and the union of a feasible set with an infeasible state equals $\bot$. When $w + w_i = W$, we mark $w_i$ as mandatory, since even if there exist alternative subsets of items in $\{i+1,\ldots,n\}$ with total weight $w_i$, these subsets can be replaced by item $i$ without loss of optimality.

For each item $\ell$ with $w_\ell \ge W/2$, we inspect the state $\mathrm{dp}(\ell + 1,w_\ell)$. If this state is feasible and nonempty, $\ell$ can be merged with an item $i$ such that $w_i \in \mathrm{dp}(\ell + 1,w_\ell)$. This is valid independently of other conditions, since such items appear together in at least one perfect packing solution.  For at least one such item $\ell$, there exist items $p$ and $q$ with $w_q \in \mathrm{dp}(\ell + 1,w_\ell)$ and $w_\ell + w_p + w_q = W$, as it holds for $\ell_h$. By Lemmas~\ref{lemma::full_pair} and~\ref{lemma::mandatory_weight}, any such triplet $(\ell,p,q)$ belongs to at least one perfect packing solution. To ensure that fixing the triplet preserves an $(\alpha,\beta,3)$-instance, we verify Theorem~\ref{theorem::partition_triplet} by checking whether a subset of items from $\widetilde{\mathcal{I}}\setminus\{\ell,p,q\}$ sums to $W$, which can be done in $O(nW)$ time by solving a knapsack problem.

If multiple triplets $(\ell,p,q)$ satisfy $w_q \in \mathrm{dp}(\ell + 1,w_\ell)$ and $w_\ell + w_p + w_q = W$, we can attempt to fix all of them without recomputing $\mathrm{dp}(i,w)$. We iterate over the triplets, and for each one, if distinct items with weights $w_\ell, w_p, w_q$ exist in the current instance and satisfy Theorem~\ref{theorem::partition_triplet}, we fix the triplet in the partial solution and remove the items from the instance.  This procedure is valid because removing items does not create new patterns. Consequently, the property that a given item $\ell$ must be packed with a subset of total weight $w$ remains true after removal, although it may no longer hold that a single item of weight $w$ exists.

Since each DP state may store up to $n$ elements, the worst-case time to compute the table is $O(n^2 W)$. For original-irreducible $(\alpha, \beta, 3)$-\textsc{ANI} instances, each state typically contains only a constant number of elements, yielding a practical runtime of $O(nW)$. For each item $\ell$, the number of identified triplets could be $O(n)$, resulting in $O(n^2)$ knapsack problems with total runtime $O(n^3 W)$. In practice, each DP computation usually identifies only a constant number of triplets $(\ell,p,q)$ satisfying Lemma~\ref{lemma::suff_property}, which may require recomputing the DP up to $O(n)$ times but reduces the number of knapsack problems per computation to a constant. Consequently, the overall worst-case complexity is $O(n^4 W)$, while the practical runtime is $O(n^2 W)$.

\subsection{A Faster Pseudopolynomial Time Algorithm for Literature ANI Instances}
\label{sec::ANI_3}
The previous algorithm already exhibits reasonable performance in practice, but it only outperforms state-of-the-art algorithms on literature \textsc{ANI} instances with around 1000 items. Next, we present a fast pseudopolynomial time algorithm for all literature \textsc{ANI} instances. To this end, we observed that only a small fraction of the $nW$ states are feasible in practice, and that literature \textsc{ANI} instances often contain many identical items. Therefore, a CSP-style algorithm can be significantly more effective.

To exploit these observations, we introduce a generalization of the DP-flow formulation, originally developed for the BPP, to the CSP, which we call the \emph{Multiplicity Flow Formulation} (MFF). Our DP algorithm will operate on the underlying graph of this formulation.

Consider an arbitrary CSP instance with a set $I = \{1,\dots,n\}$, where each item type $i \in I$ has width $w_i$ and demand $d_i$. Without loss of generality, we assume that the items have distinct weights and they are sorted in decreasing order of weight.

\paragraph{Multiplicity Flow graph.}
We define a directed acyclic graph $G = (\mathcal{V},\mathcal{A})$ with vertex set
\[
	\mathcal{V} = \{(i,w) \mid i = 1,\dots,n,\; w = 0,\dots,W\} \cup \{(n+1,W)\}.
\]
The vertex $s = (1,0)$ is the \emph{source node}, and the vertex $t = (n+1,W)$ is the \emph{sink node}. Each vertex $(i,w)$ represents a
partial solution obtained after considering item types $1,\dots,i-1$, with total used capacity equal to $w$.

For each item type $i = 1,\dots,n$, for each vertex $(i,w) \in \mathcal{V}$, and for each multiplicity $m \in \{0,1,\dots,d_i\}$, there is an arc $(u,v,m) \in \mathcal{A}$, where
\[
	u = (i,w) \quad \text{and} \quad v = (i+1,\, w + m w_i),
\]
whenever $v \in \mathcal{V}$. The value $m$ represents the number of copies of item type $i$ selected at this stage. Arcs with $m = 0$ are referred to as \emph{bypass arcs}.

Each directed path from the source $(1,0)$ to the sink $(n+1,W)$ corresponds to a feasible cutting pattern $P$ that exactly fills a bin of capacity $W$. Hence, the arc set represents only full patterns. Although this construction could be extended to allow patterns with waste, they are not necessary for our purposes.

The resulting MFF is defined as follows:
\begingroup
\small
\begin{alignat}{3}
	\mathrm{(MFF)}\enspace & \min \enspace    &   & z \\
	                  & \text{s.t.: } &   &
	\sum_{(u,v,m)\in\mathcal{A}} \varphi_{(u,v,m)}
	-
		\sum_{(v,u,m)\in\mathcal{A}} \varphi_{(v,u,m)}
	=
	\begin{cases}
	-z & \text{if } v = s,\\
	\;\;z & \text{if } v = t,\\
	0 & \text{otherwise},
	\end{cases}
	&& \forall v \in \mathcal{V}, \label{dp::cons} \\
	                  &               &   &
	\sum_{(u,v,m)\in\mathcal{A}_i} m \cdot \varphi_{(u,v,m)} = d_i,
	&& \forall i \in I, \label{dp::dem} \\
	                  &               &   &
	\varphi_{(u,v,m)} \in \mathbb{Z}^+, \quad \forall (u,v,m) \in \mathcal{A},
	&& \\
	                  &               &   &
	z \in \mathbb{Z}^+ .
\end{alignat}
\endgroup

The variables $\varphi_{(u,v,m)}$ represent the amount of flow on each arc $(u,v,m)$, $\mathcal{A}_i$ denotes the subset of arcs associated with item type $i$, and $z$ represents the total flow value. Constraints~\eqref{dp::cons} enforce flow conservation, while constraints~\eqref{dp::dem} ensure that the demand $d_i$ of each $i$ is satisfied exactly.

Given an $(\alpha,\beta,3)$-\textsc{ANI} instance $\widetilde{\mathcal{I}}$, our algorithm performs a DP technique similar to the one presented in Section~\ref{sec::ANI_2}, but applied to the underlying graph of the MFF. Algorithm~\ref{alg::build_graph} presents an efficient method to compute the arc sets $\mathcal{A}_i$ without bypass arcs (which are unnecessary, as explained later). Consider an arc $(u, v, m) \in \mathcal{A}_i$, where $u = (i, w)$ and $v = (i + 1, w + m w_i)$. For simplicity, we also denote it by $(y, m)$, where $y = w$, since for each $(y, m) \in \mathcal{A}_i$ the vertices $u$ and $v$ are uniquely determined.

The algorithm first computes $\mathcal{A}^\mathrm{f}_i$, the set of arcs from item type $i$ belonging to ordered paths starting from the source node, where an ordered path is one in which items appear in decreasing weight. For each item type $i$, we maintain a forward DP vector $\mathrm{P}_f$ that records the reachable partial capacities using items $\{1,\ldots,i-1\}$.  Next, a backward search computes $\mathcal{A}^\mathrm{b}_i$, consisting of all arcs in $\mathcal{A}^\mathrm{f}_i$ that can also reach the sink node via an ordered path. This backward phase processes items in reverse order using a DP vector $\mathrm{P}_b$, representing partial capacities reachable using item types $\{i+1,\ldots,n\}$. Finally, $\mathcal{A}_i$ is equal to $\mathcal{A}^\mathrm{b}_i$.

\begin{algorithm}[t]
	\setstretch{1}
	\caption{Compute Adjacency Graph of the Multiplicity--Flow Network}
	\label{alg::build_graph}
	\KwIn{Distinct weights $w_1,\dots,w_n$ in decreasing order;\\
        demands $d_1,\dots,d_n$ and capacity $W$}
	\KwOut{For each $i$, a set $\emph{Adj}[i]$ of arcs $(y,m)$}
	Initialize $\emph{Adj}[i] \gets \emptyset$ for all $i \in \{1, \ldots, n\}$\;
	Initialize forward DP array $\mathrm{P}_f[0 \dots W] \gets \mathrm{false}$\;
	$\mathrm{P}_f[0] \gets \mathrm{true}$\;
	\For{$i \gets 1$ \KwTo $n$}{
		$w \gets w_i$\;
		\For{$y \gets W - w$ \KwDownTo $0$}{
			\For{$m \gets 1$ \KwTo $d_i$}{
				\If{$y + m \cdot w \le W$ \textbf{and} $\mathrm{P}_f[y]$}{
					Add $(y,m)$ to $\emph{Adj}[i]$\;
					$\mathrm{P}_f[y + m \cdot w] \gets \mathrm{true}$\;
				}
			}
		}
		Reverse the order of $\emph{Adj}[i]$\;
	}
	Initialize backward DP array $\mathrm{P}_b[0 \dots W] \gets \mathrm{false}$\;
	$\mathrm{P}_b[W] \gets \mathrm{true}$\;
	\For{$i \gets n$ \KwDownTo $1$}{
		$w \gets w_i$\;
		$\mathrm{NewAdj} \gets \emptyset$\;
		\ForEach{$(y,m) \in \emph{Adj}[i]$}{
			\If{$\mathrm{P}_b[y + m \cdot w]$}{
				Add $(y,m)$ to $\mathrm{NewAdj}$\;
				$\mathrm{P}_b[y] \gets \mathrm{true}$\;
			}
		}
		$\emph{Adj}[i] \gets \mathrm{NewAdj}$\;
	}
	\Return $\emph{Adj}$\;
\end{algorithm}

The sets $\mathcal{A}_i$ are represented by the vector $\emph{Adj}[i]$ and are constructed efficiently by Algorithm~\ref{alg::build_graph}, which implements the procedure described above without explicitly referring to $\mathcal{A}_i^{\mathrm{f}}$ and $\mathcal{A}_i^{\mathrm{b}}$. Algorithm~\ref{alg::dp_on_graph} presents the pseudocode of the DP procedure executed over the adjacency graph $\emph{Adj}$, defined by the collection of vectors $\emph{Adj}[i]$. The algorithm processes the item types in reverse order. For each item type $i$, it extends the DP vector computed for item type $i+1$ to obtain its corresponding vector.

\begin{algorithm}[t]
	\setstretch{1}
	\caption{DP over Adjacency Graph for Mandatory Item Detection}
	\label{alg::dp_on_graph}
	\KwIn{
		Adjacency lists $\mathrm{Adj}[i]$ for items $i=1,\dots,n$; \\
		weights $w_i$, demands $d_i$, capacity $W$
	}
	\KwOut{
		For each partial capacity $y$, a set $\mathrm{dp}[y]$ of mandatory weights\;
	}
	Initialize $\mathrm{dp}[y] \gets \bot$ for all $y = 0,\dots,W$\;
	$\mathrm{dp}[W] \gets \emptyset$\;
	\For{$i \gets n$ \KwDownTo $1$}{
		\textbf{if} $d_i = 0$ \textbf{then} \textbf{continue}\;
		\ForEach{$(y,m) \in \emph{Adj}[i]$}{
			\If{$m \le d_i$ \textbf{and} $\mathrm{dp}[y + m \cdot w_i] \neq \bot$}{
				\textbf{if} $\mathrm{dp}[y] = \bot$ \textbf{or} $y + w_i = W$ \textbf{then}\\
                \hspace*{1.5em}$\mathrm{dp}[y] \gets \mathrm{dp}[y + m \cdot w_i] \cup \{w_i\}$\;
				\textbf{else}\\
                \hspace*{1.5em}$\mathrm{dp}[y] \gets \mathrm{dp}[y] \cap (\mathrm{dp}[y + m \cdot w_i] \cup \{w_i\})$\;
			}
		}
	}
	\Return $\mathrm{dp}$\;
\end{algorithm}

Since the arcs $(y,m) \in \emph{Adj}[i]$ are stored in non-decreasing order of the position $y$, the DP vector can be safely updated \emph{in place} by iterating over $\emph{Adj}[i]$. This in-place update eliminates the need for bypass arcs: when computing the DP for item type $i$, the DP vector is already initialized with the values corresponding to item type $i+1$, and only the transitions associated with $i$ must be applied.

\begin{lemma}\label{lemma::base_case}
	Let $\mathrm{dp}^i$ denote the vector $\mathrm{dp}$ after processing item type $i$ in Algorithm~\ref{alg::dp_on_graph}. If item type $i$ corresponds to some item $\ell_k$ for $k \in \{1,\ldots,h\}$, then $\mathrm{dp}^1[w_{\ell_k}] = \mathrm{dp}^{i+1}[w_{\ell_k}]$.
\end{lemma}

\begin{proof}
	Choose an item type $i$ corresponding to some $\ell_k$. First, suppose that $w_{\ell_k} > W/2$. By ordering, for all item types $i' \in \{1, \ldots, i\}$, there is no arc $(w_{\ell_k}, w_{\ell_k}+w_{i'})$, since $w_{i'} \geq w_{\ell_k} \rightarrow w_{i'} + w_{\ell_k} > W$. Hence, the value at position $w_{\ell_k}$ cannot be updated after processing $i + 1$, and therefore $\mathrm{dp}^1[w_{\ell_k}] = \mathrm{dp}^{i+1}[w_{\ell_k}]$.

    Now suppose that \( w_{\ell_k} = W/2 \). Then \( \ell_k \) is the unique item of weight \( W/2 \) in the BPP instance, that is, \( d_i = 1 \); otherwise, Condition~\ref{cond:no-completion} would be violated. Hence, when processing item \(i\), there does not exist an arc from capacity \(W/2\) to capacity \(W\), since this arc can only be generated by an ordered path consisting of an arc from \(0\) to \(W/2\) followed by another from \(W/2\) to \(W\), which requires \(d_i \geq 2\). Therefore, \( \mathrm{dp}^1[w_{\ell_k}] = \mathrm{dp}^{i+1}[w_{\ell_k}] \).
\end{proof}

By Lemma~\ref{lemma::base_case}, and since it suffices to fix the items $\ell_k$, we can apply the fixing strategy described in Section~\ref{sec::ANI_2} to each item type $\ell$ with $w_\ell \ge W/2$ using the value $\mathrm{dp}[w_\ell]$. If a triplet $(\ell,p,q)$ satisfying Lemma~\ref{lemma::suff_property} is identified, we check whether it satisfies Theorem~\ref{theorem::partition_triplet} by solving a knapsack problem over the underlying graph defined by $\emph{Adj}$. To this end, we ignore the arcs in $\emph{Adj}[p]$ and $\emph{Adj}[q]$ and check whether a path from capacity $0$ to $w_\ell$ contains an arc associated with item $\ell$.

Algorithm~\ref{alg::build_graph} runs in \(O\!\left(\sum_{i \in I} d_i W\right)\) time and is executed only once. Algorithm~\ref{alg::dp_on_graph} performs \(O(|\emph{Adj}| + W)\) iterations, where $|\emph{Adj}|$ denotes the number of edges in the adjacency graph. Each iteration typically runs in $O(1)$ time in practice, since the sets $\mathrm{dp}[y]$ are usually of constant size. The same complexity \(O\!\left(\sum_{i \in I} d_i W\right)\) applies to each knapsack problem solved. Consequently, this approach is significantly faster than recomputing the DP or solving knapsack problems from scratch, which would require \(O\!\left(\sum_{i \in I} d_i W\right)\) time.

The performance can be further improved by periodically pruning arcs that are no longer valid, either because item demands have decreased or because the corresponding arcs no longer belong to any full pattern after the item removals. This pruning step can be implemented using a procedure similar to Algorithm~\ref{alg::build_graph}, but with a runtime of \(O(|\emph{Adj}| + W)\), since the existing adjacency graph \(\emph{Adj}\) can be used directly as input.

For an even faster algorithm, we may assume that every triplet $(\ell,p,q)$ can be fixed. This does not affect the correctness of the reductions, since any identified triplet can indeed be fixed while preserving at least one perfect packing solution, if one exists. However, this removes the guarantee that the instance is reduced to size $\alpha$ before no further triplets can be fixed, as intermediate instances may no longer be $(\alpha,\beta,3)$-\textsc{ANI}. Nevertheless, for all instances reported in the literature, even when this step is skipped, the algorithm consistently reduces the instances to at most 15 items.

Finally, all versions of the algorithm can be adapted to the SSP with minimal modifications. In this setting, the backtracking phase is used to certify that no solution using $\beta$ bins exists for the original instance $\widetilde{O}$, and to identify $\beta - 1$ bins whose total occupancy is at least the capacity $W$.

\section{Computational Experiments}~\label{sec::experiments}
Next, we present our computational experiments for the BPP and CSP, comparing our two solvers with other state-of-the-art algorithms. Our solvers were implemented in C++ and can be found at \url{https://gitlab.com/renanfernandofranco/ai-and-ani-solvers}. The results for our algorithms were obtained on a computer running Ubuntu 22.04.1 LTS (64-bit), using C++14 compiled with GCC 11.4.0, and an Intel\textsuperscript{\textregistered} Xeon\textsuperscript{\textregistered} CPU E5-2630 v4 @ 2.20 GHz with 64 GB of RAM. This processor has a single-thread PassMark indicator (STPI) of 1785. We consider STPI in subsequent sections to compare CPU performance, as available at \href{https://www.passmark.com}{https://www.passmark.com}, where higher values indicate better performance.

We refer to the algorithms introduced in Section~\ref{sec::AI_3} and Section~\ref{sec::ANI_3} as the \textsc{AI} and \textsc{ANI} solvers, respectively. We evaluate their performance against three state-of-the-art solvers: \citet{Loti_2023}, \citet{Baldacci_2024}, and \citet{Silva_2025}.
To ensure a fair comparison, all solvers were executed with a maximum time limit of one hour per instance. The computational environments used in the literature are as follows:
\begin{itemize}
    \item \citet{Loti_2023}: Intel Xeon E3-1245 v5 @ 3.50 GHz (STPI 2249) with 32 GB of RAM;
    \item \citet{Baldacci_2024}: Intel Xeon Gold 6130 @ 2.10 GHz (STPI 2067) with 192 GB of RAM;
    \item \citet{Silva_2025}: Executed on the same hardware used in this study.
\end{itemize}

Although it would be possible to include a comparison with state-of-the-art SSP solvers~\citep{Silva_2025}, we omit it here, as these solvers perform strictly worse than the corresponding BPP/CSP ones.

Table~\ref{tab:results} presents the computational results obtained on the AI and ANI classes. The columns \texttt{Opt} and \texttt{Avg Time} indicate, respectively, the number of instances solved to proven optimality and the average runtime in seconds for each group $XY$ of instances, with $X$ indicating the class and $Y$ the number of items. \citet{Baldacci_2024} present experiments only for the \textsc{ANI} class. For the AI/ANI solver columns, the reported times for each class correspond to the solver designed for that class. Best values for each line are highlighted in bold.

\begin{table}[!htbp]
	\centering
\setlength{\tabcolsep}{3pt}
	\caption{Comparison between AI/ANI solvers and state-of-the-art algorithms for the AI and ANI classes.}
	\label{tab:results}
	\small
		\begin{tabular}{l r rr rr rr rr}
			\toprule
			Class & Total 
			& \multicolumn{2}{c}{de Lima et al.} 
			& \multicolumn{2}{c}{Baldacci et al.} 
			& \multicolumn{2}{c}{\makecell{Silva and\\Schouery}}
			& \multicolumn{2}{c}{AI/ANI solvers} \\
			\cmidrule(lr){3-4} 
			\cmidrule(lr){5-6} 
			\cmidrule(lr){7-8} 
			\cmidrule(lr){9-10}
			& & Opt & Avg Time 
			& Opt & Avg Time 
			& Opt & Avg Time 
			& Opt & Avg Time \\
			\midrule
			AI 202         & 50  & \textbf{50}  & 2.0    & --  & --     & \textbf{50}  & 0.4   & \textbf{50}  & \textbf{0.09} \\
			AI 403         & 50  & \textbf{50}  & 25.2   & --  & --     & \textbf{50}  & 5.0   & \textbf{50}  & \textbf{0.12} \\
			AI 601         & 50  & 49  & 50.0   & --  & --     & \textbf{50}  & 57.1  & \textbf{50}  & \textbf{1.01} \\
			AI 802         & 50  & 46  & 556.5  & --  & --     & 48  & 223.7 & \textbf{50}  & \textbf{0.06} \\
			AI 1003        & 50  & 36  & 1577.1 & --  & --     & 42  & 794.9 & \textbf{50}  & \textbf{0.05} \\
			\midrule
			ANI 201        & 50  & \textbf{50}  & 3.0    & \textbf{50}  & 13.6   & \textbf{50}  & 0.4   & \textbf{50}  & \textbf{0.02} \\
			ANI 402        & 50  & \textbf{50}  & 24.9   & \textbf{50}  & 308.2  & \textbf{50}  & 2.2   & \textbf{50}  & \textbf{0.25} \\
			ANI 600        & 50  & \textbf{50}  & 140.7  & 25  & 1931.5 & \textbf{50}  & 11.8  & \textbf{50}  & \textbf{1.01} \\
			ANI 801        & 50  & 49  & 393.2  & 3   & 3352.7 & \textbf{50}  & 57.0  & \textbf{50}  & \textbf{3.30} \\
			ANI 1002       & 50  & 43  & 1302.5 & 0   & 3600.0 & 47  & 374.0 & \textbf{50}  & \textbf{8.89} \\
			\midrule
			\textbf{Total} & 500 & 473 & --     & 178 & --     & 487 & --    & \textbf{500} & --   \\
			\bottomrule
		\end{tabular}
\end{table}

Observe that our algorithms are approximately one order of magnitude faster, particularly for the \textsc{AI} class. Additionally, our solvers successfully solve all instances to proven optimality in both classes, in contrast to the previous state-of-the-art algorithm~\citep{Silva_2025} that leaves 13 unproven instances.

To better assess the performance of our algorithms, we conducted extensive experiments, including additional instance classes, to evaluate their computational behavior across different scenarios.
Our algorithms require that $\sum_{i \in I} w_i / W = K^I \in \mathbb{Z}_+$ and that at least $K^I-3$ items have weights $w_i \geq W/2$.
Instances satisfying both conditions are referred to as \textit{eligible}.
Accordingly, our solvers first check eligibility; if an instance does not satisfy these conditions, execution terminates immediately.

In addition to the \textsc{AI} and \textsc{ANI} classes, we consider all instances from BPPLib \citep{Delorme_2016}.
We also evaluate the challenging \texttt{SSTriplets} instances by \citet{Silva_2025}. All instances turned out not to be eligible, with a few exceptions for the classes Random (10 eligible instances) and Scholl (3 eligible ones). Consequently, we report computational results only for these two classes. Table~\ref{tab:detailed_AI_fast} presents detailed results for the \textsc{AI} solver across all aforementioned instances without imposing a time limit. For each class, we report the number of instances, the number of eligible instances, the number of instances solved to proven optimality (\textit{Opt}), the average run time in seconds (\textit{Avg Time}), the maximum run time in seconds (\textit{Max Time}), and the average number of recursive calls (\textit{Avg Rec Calls}). We also report the number of instances that reached the base case (i.e., a recursive call with $h$ fixed triplets) at least once (\textit{Reach Base}), indicating how many instances were able to fix $h$ triplets at least once, along with the average number of base cases explored (\textit{Avg Base Cases}).

\begin{table}[!htbp]
\centering
\setlength{\tabcolsep}{3pt}
\caption{Detailed performance of the \textsc{AI} solver across benchmark classes.}
\label{tab:detailed_AI_fast}
\small
\begin{tabular}{lrrrrrrrr}
\toprule
\makecell{Class}
& \makecell{Total}
& \makecell{Eligible}
& \makecell{Opt}
& \makecell{Avg\\Time\\(s)}
& \makecell{Max\\Time\\(s)}
& \makecell{Avg\\Rec\\Calls}
& \makecell{Reach\\Base}
& \makecell{Avg\\Base\\Cases} \\
\midrule
\textsc{AI 202}
& 50 & 50 & 50 & 0.09 & 2.16 & 6.6e4 & 1.0 & 1.90 \\
\textsc{AI 403}
& 50 & 50 & 50 & 0.12 & 3.99 & 6.8e4 & 1.0 & 3.24 \\
\textsc{AI 601}
& 50 & 50 & 50 & 1.01 & 47.62 & 4.4e5 & 1.0 & 1.18 \\
\textsc{AI 802}
& 50 & 50 & 50 & 0.06 & 0.47 & 2.4e4 & 1.0 & 9.42 \\
\textsc{AI 1003}
& 50 & 50 & 50 & 0.05 & 1.43 & 1.1e4 & 1.0 & 1.90 \\
\midrule
\textsc{ANI 201}
& 50 & 50 & 0 & 3.96 & 8.80 & 2.8e6 & 1.0 & 60.26 \\
\textsc{ANI 402}
& 50 & 50 & 0 & 18.04 & 78.10 & 8.6e6 & 1.0 & 58.66 \\
\textsc{ANI 600}
& 50 & 50 & 0 & 60.33 & 257.78 & 2.1e7 & 1.0 & 20.74 \\
\textsc{ANI 801}
& 50 & 50 & 0 & 154.48 & 1207.47 & 4.2e7 & 1.0 & 45.00 \\
\textsc{ANI 1002}
& 50 & 50 & 0 & 164.40 & 4227.55 & 3.7e7 & 1.0 & 13.52 \\
\midrule
\textsc{Random}
& 3840 & 10 & 0 & 0.00 & 0.00 & 0.0 & 0.0 & 0.00 \\
\textsc{Scholl}
& 1200 & 3 & 0 & 0.00 & 0.00 & 0.0 & 0.0 & 0.00 \\
\bottomrule
\end{tabular}
\end{table}

As expected, only the \textsc{AI} instances are solved consistently. Their runtimes are substantially smaller than those of the \textsc{ANI} instances, mainly because the latter require two orders of magnitude more recursive calls on average. Although some \textsc{AI} instances exhibit runtimes well above the class average, all of them are solved within 48 seconds, whereas the maximum runtime for the \textsc{ANI} class reaches 4227 seconds. For both classes, the algorithm is always able to recover \(h\) triplets at least once (Reach Base is equal to 1.0), and multiple sets of \(h\) triplets may satisfy the fixing conditions, as shown by an average number of base cases greater than 1. Observe that substantially fewer base cases are explored for \textsc{AI} instances. This suggests that, for these instances, a residual instance packable into three bins is typically encountered after exploring only a fraction of the possible sets of \(h\) triplets that can be fixed. One possible explanation is that these instances admit several optimal solutions. In contrast, \textsc{ANI} instances require exploring every candidate set satisfying the pruning conditions.

In the worst case, the number of recursive calls is $O(n^6)$, corresponding to $O(n^5)$ candidate sets times $h$ iterations (Section~\ref{sec::AI_2}). For \textsc{AI} instances, the observed number of recursive calls appears bounded by $O(n^2)$, indicating that only $O(n)$ candidates are evaluated on average. If we consider that each recursive call takes time complexity $O(n)$ on average, this yields an apparent performance of $O(n^3)$. For \textsc{ANI} instances, performance appears to be bounded in practice by $O(n^4)$.

Additionally, none of the eligible \textsc{Random} and \textsc{Scholl} instances could be solved, and the solver terminates in less than one millisecond on average for these classes. The average number of recursive calls for such instances is so small that, when rounded to one decimal place, it appears as zero.

Regarding the \textsc{ANI} solver, we next report experiments in which we modify it to search for a perfect packing for residual instances $\widehat{I}$ whenever $K^{\widehat{I}} = \sum_{i \in \widehat{I}} w_i / W \le 5$. This modification is unnecessary for \textsc{ANI} instances, since the solver can always recover all original triplets, so that the residual instances correspond to the original ones. However, the larger threshold allows our algorithm to solve several additional \textsc{AI} instances in which the resulting residual instance is slightly larger than the original, a situation that occurs frequently. Note that the fixing strategy of the \textsc{ANI} solver is optimal only if the residual instance $\widehat{I}$ can be packed into at most $K^{\widehat{I}}+1$ bins. Therefore, an instance is declared solved by the \textsc{ANI} solver if $K^{\widehat{I}} \le 5$ and the optimal solution to $\widehat{I}$ obtained by backtracking uses at most $K^{\widehat{I}}+1$ bins; otherwise, it is classified as unsolved.

Table~\ref{tab:detailed_ANI} reports the performance of the \textsc{ANI} solver on these instances. For each class, we provide the total number of instances, the number of eligible instances, average execution time in seconds, and maximum execution time in seconds. For eligible instances, we also report the number of instances solved to proven optimality, the average number of iterations (\textit{Avg Iter.}), which corresponds to the average number of times that Algorithm~\ref{alg::dp_on_graph} is executed, and the average number of \textit{fixed triplets} (note that more than one triplet can be fixed per iteration). The column \textit{Avg Residual Size} gives the average size of the instance remaining after all fixing procedures. Additionally, whenever Algorithm~\ref{alg::dp_on_graph} is executed to identify a triplet, we compute the ratio $W \sum_{i \in \widehat{I}} w_i d_i / |\emph{Adj}|$, where $\widehat{I}$ is the set of remaining items and $|\emph{Adj}|$ is the number of edges in the network. This metric captures the relative speedup of our approach compared to recomputing the DP from scratch, as discussed in Section~\ref{sec::ANI_2}. The average of this ratio, aggregated by class, is reported in the column \textit{Avg. DP Size/$|\emph{Adj}|$}.

\begin{table}[!htbp]
\centering
\setlength{\tabcolsep}{3pt}
\caption{Detailed performance of the ANI solver across benchmark classes.}
\label{tab:detailed_ANI}
\small
    \begin{tabular}{l r r r r r r r r r}
        \toprule
        Class
        & Total
        & Eligible
        & Opt
        & \makecell{Avg. \\ Time \\ (s)}
        & \makecell{Max. \\ Time \\ (s)}
        & \makecell{Avg. \\ Iter.}
        & \makecell{Avg. \\ Fixed \\ triplets}
        & \makecell{Avg. \\ Residual \\ Size}
        & \makecell{Avg. \\ DP Size\\ $/|\emph{Adj}|$ } \\
        \midrule
        AI 202   & 50 & 50 & 12 & 0.02 & 0.03 & 53.75  & 61.83  & 16.50 & 15.03 \\
        AI 403   & 50 & 50 & 27 & 0.27 & 0.36 & 110.00 & 128.19 & 18.44 & 25.20 \\
        AI 601   & 50 & 50 & 27 & 1.01 & 1.42 & 161.41 & 194.11 & 18.67 & 33.62 \\
        AI 802   & 50 & 50 & 27 & 3.37 & 4.87 & 215.19 & 261.37 & 17.89 & 50.68 \\
        AI 1003  & 50 & 50 & 30 & 8.72 & 11.34 & 265.17 & 327.93 & 19.20 & 75.86 \\
        \midrule
        ANI 201  & 50 & 50 & 50 & 0.02 & 0.03 & 54.30  & 62.00  & 15.00 & 14.95 \\
        ANI 402  & 50 & 50 & 50 & 0.25 & 0.37 & 111.12 & 129.00 & 15.00 & 24.82 \\
        ANI 600  & 50 & 50 & 50 & 1.01 & 1.44 & 163.34 & 195.00 & 15.00 & 35.18 \\
        ANI 801  & 50 & 50 & 50 & 3.30 & 4.92 & 216.10 & 262.00 & 15.00 & 50.69 \\
        ANI 1002 & 50 & 50 & 50 & 8.89 & 13.74 & 265.84 & 329.00 & 15.00 & 76.01 \\
        \midrule
        Random         & 3840 & 10 & 0 & 0.00 & 0.00 & 0.00 & 0.00 & 0.00 & 0.00 \\
        Scholl         & 1200 & 3  & 0 & 0.00 & 0.00 & 0.00 & 0.00 & 0.00 & 0.00 \\
        \bottomrule
    \end{tabular}
\end{table}

We observe that the \textsc{ANI} solver is only able to solve instances from the \textsc{AI} and \textsc{ANI} classes. Although some instances from other classes, such as \textsc{Random} and \textsc{Scholl}, are eligible, they do not satisfy the conditions required to be solved. Within the \textsc{AI} class, a substantial number of instances are solved, with more than half of the instances solved for classes with 403 items or more. The average number of iterations is very close to the average number of fixed triplets, indicating that typically one triplet is fixed per execution of Algorithm~\ref{alg::dp_on_graph}. As expected, the residual size for \textsc{ANI} instances is consistently equal to 15. For \textsc{AI} instances, the residual size is larger but still relatively small, although not always small enough to allow backtracking to solve all cases. Finally, the number of edges in the network is significantly smaller than the size of the DP computed from scratch, being about 15 times smaller on average for instances with 200 items and around 75 times smaller for instances with 1000 items.

Finally, for the \textsc{ANI} solver, the average and maximum runtimes increase progressively with the number of items in both the \textsc{AI} and \textsc{ANI} classes, indicating a more consistent behavior than the \textsc{AI} solver, with the maximum runtime remaining less than twice the average runtime. The maximum runtime is only 13.7 seconds.

\section{Conclusions}\label{sec::conclusion}
We investigated the Bin Packing Problem (BPP), the Cutting Stock Problem (CSP) and the Skiving Stock Problem (SSP), with a particular focus on the Augmented Integer Round Up Property (AI) and Augmented Non-Integer Round Up Property (ANI) benchmark instances proposed by Delorme et al. (2016). 

Devising BPP, CSP and SSP instances that are both reasonably sized and computationally challenging has become increasingly difficult with the development of strong formulations and advanced solution methods. The AI and ANI instances were introduced to serve as a playground for new solution methods, and, indeed, significant progress has been achieved since their proposal, with several new exact algorithms being proposed in the literature. Despite that, state-of-the-art algorithms (culminating in da Silva and Schouery, 2025) are not able to solve all such instances within a one-hour time limit on a standard computer. 

In this paper, we studied the combinatorial structure of the AI and ANI instances and introduced generalized versions of their construction procedure. We established structural properties that helped us devise efficient algorithms. In particular, we proposed polynomial-time algorithms for the AI class and pseudopolynomial-time algorithms for the ANI class. 

Computational experiments showed that proposed algorithms solve all AI and ANI instances to proven optimality in a matter of seconds. In addition, they can be used as a quick preprocessing routines in exact methods, although they are guaranteed to return an optimality status only for instances belonging to the class for which they were designed.
These findings indicate that new instances are needed for these problems. Recent work has already taken steps in this direction. For example, da Silva and Schouery (2025) proposed the SSTriplets benchmark, in which the optimal solution is a perfect packing of triplets. These instances are substantially harder than earlier benchmarks, and some instances with up to 216 items could not be solved within a one-hour time limit by state-of-the-art algorithms. 

Overall, there is still scope to design new benchmarks that are both challenging and practical. To this aim, the key theoretical question is still the Modified Integer Round-Up Property (MIRUP) conjecture, which states that the integer optimum exceeds the ceiling of the linear relaxation of the Set Covering Formulation by at most one bin. A recent preprint claims to disprove this conjecture by constructing instances with unbounded additive gaps \citep{OpenAI2026BinPacking}. Its construction, however, produces an enormous number of items (roughly $10^{20+}$), making the instances impractical as explicit computational benchmarks. This leaves considerable interest in developing BPP and CSP instances that are both difficult and of manageable size.

\section*{Statements and Declarations}

\subsection*{Funding}

This work was supported by the National Council for Scientific and Technological Development (CNPq), under grant numbers 312345/2023-2 and 404779/2025-5, and the São Paulo Research Foundation (FAPESP), under grant numbers 2022/05803-3, 2023/17964-4, and 2024/17437-7. This work was also supported by the PNRR-TNE project \emph{Green \& Pink for Sustainable Education (GPSEducation)} (Project Proposal TNE23-00012, CUP D74G23000280006), funded by the European Union – NextGenerationEU under the Italian National Recovery and Resilience Plan (PNRR), Mission 4 \emph{Education and Research}, Component 1.

\subsection*{Competing interests}

The authors declare that they have no conflicts of interest.

\subsection*{Data and code availability}

The source code, benchmark instances, and experimental results used in this study are publicly available in the GitLab repository at \url{https://gitlab.com/renanfernandofranco/ai-and-ani-solvers}.



\end{document}